\documentclass[11pt,letterpaper]{article}

\usepackage{etoolbox}

\title{Reducing Path TSP to TSP} 

\author{
Vera Traub\thanks{
Research Institute for Discrete Mathematics and Hausdorff Center for Mathematics, University of Bonn, Bonn, Germany.
Email: \href{mailto:traub@or.uni-bonn.de}%
{traub@or.uni-bonn.de}, \href{mailto:vygen@or.uni-bonn.de}%
{vygen@or.uni-bonn.de}.
This research was initiated while the first two authors visited FIM at ETH Zurich.
}
\and
Jens Vygen\footnotemark[1]
\and
Rico Zenklusen\thanks{
Department of Mathematics, ETH Zurich, Zurich, Switzerland.
Email: \href{mailto:ricoz@math.ethz.ch}%
{ricoz@math.ethz.ch}.
Research supported in part by the Swiss National Science Foundation grants 200021\_184622 and 200021\_165866.}
}

\date{July 24, 2019}

\usepackage[margin=1in]{geometry}

\usepackage{amsmath}
\usepackage{amssymb}
\usepackage{amsthm}
\usepackage{thm-restate}
\newtheorem{theorem}{Theorem}
\newtheorem{lemma}[theorem]{Lemma}

\newtheorem{definition}[theorem]{Definition}

\newtheorem{corollary}[theorem]{Corollary}

\newtheorem{claim}[theorem]{Claim}

\usepackage{url}

\usepackage[table]{xcolor}
\usepackage{array}

\usepackage{setspace}

\usepackage{wrapfig}

\usepackage[pdfpagelabels,bookmarks=false]{hyperref}
\hypersetup{colorlinks, linkcolor=darkblue, citecolor=darkgreen, urlcolor=darkblue}

\makeatletter
\newcommand{\labeltarget}[1]{\Hy@raisedlink{\hypertarget{#1}{}}}
\makeatother

\usepackage{times}

\usepackage[inline]{enumitem}
\usepackage{moreenum}

\usepackage{mathtools}
\usepackage[super]{nth}

\setlist[enumerate]{nosep,topsep=0em}
\setlist[enumerate,1]{label=(\roman*), leftmargin=2.2em}
\setlist[enumerate,2]{label=(\alph*)}
\setlist[itemize]{nosep,topsep=0.1em}

\usepackage[textsize=footnotesize, color=blue!30!white]{todonotes}

\usepackage{xfrac}

\usepackage{tikz}
\usetikzlibrary{calc}
\usetikzlibrary{shapes.geometric}
\usetikzlibrary{arrows}
\usetikzlibrary{decorations.pathreplacing, decorations.pathmorphing}

\usepackage{tcolorbox}
\tcbuselibrary{skins}

\usepackage{float}
\usepackage{graphicx}
\graphicspath{{graphics/}}
\makeatletter
\newcommand\appendtographicspath[1]{%
  \g@addto@macro\Ginput@path{#1}%
}
\makeatother

\makeatletter
\let\@fnsymbol\@arabic
\makeatother

\usepackage[margin=10pt,font=small,labelfont=bf]{caption}
\usepackage{subcaption}

\usepackage[vlined,ruled,algo2e]{algorithm2e}
\SetAlCapHSkip{0.2em}
\SetCustomAlgoRuledWidth{0.95\textwidth}
\makeatletter
\patchcmd{\@algocf@start}{%
  \begin{lrbox}{\algocf@algobox}%
}{%
  \rule{0.025\textwidth}{\z@}%
  \begin{lrbox}{\algocf@algobox}%
  \begin{minipage}{0.95\textwidth}%
}{}{}
\patchcmd{\@algocf@finish}{%
  \end{lrbox}%
}{%
  \end{minipage}%
  \end{lrbox}%
}{}{}
\makeatother

\usepackage{mdframed}

\definecolor{darkblue}{rgb}{0,0,0.38}
\definecolor{darkred}{rgb}{0.6,0,0}
\definecolor{darkgreen}{rgb}{0.1,0.35,0}

\DeclareMathOperator{\width}{width}

\newcommand{\symdiff}{\bigtriangleup}
\newcommand\OPT{\ensuremath{\mathrm{OPT}}}
\newcommand\TSP{\ensuremath{\mathrm{TSP}}}
\renewcommand{\epsilon}{\varepsilon}
\newcommand{\odd}{\ensuremath{\mathrm{odd}}}

\newcommand\APX{\ensuremath{\mathsf{APX}}}

\def\Ascr{\mathcal{A}}
\def\Bscr{\mathcal{B}}
\def\Cscr{\mathcal{C}}

\def\Lscr{\mathcal{L}}

\def\cupp{\stackrel{.}{\cup}}

\makeatletter
\let\@@pmod\pmod
\DeclareRobustCommand{\pmod}{\@ifstar\@pmods\@@pmod}
\def\@pmods#1{\mkern8mu({\operator@font mod}\mkern 6mu#1)}
\makeatother

\makeatletter
\let\@@mod\mod
\DeclareRobustCommand{\mod}{\@ifstar\@mods\@@mod}
\def\@mods#1{\mkern8mu{\operator@font mod}\mkern 6mu#1}
\makeatother

\begin{document}

\maketitle
\thispagestyle{empty}
\addtocounter{page}{-1}

\begin{abstract}
We present a black-box reduction from the path version of the Traveling Salesman Problem (Path TSP) to the classical tour version (TSP). More precisely, we show that given an $\alpha$-approximation algorithm for TSP, then, for any $\epsilon >0$, there is an 
$(\alpha+\epsilon)$-approximation algorithm for the more general Path TSP. 
This reduction implies that the approximability of Path TSP is the same as for TSP, up to an arbitrarily small error. This avoids future discrepancies between the best known approximation factors achievable for these two problems, as they have existed until very recently. 

A well-studied special case of TSP, Graph TSP, asks for tours in unit-weight graphs.
Our reduction shows that any $\alpha$-approximation algorithm for Graph TSP implies an $(\alpha+\epsilon)$-approximation algorithm for its 
path version.
By applying our reduction to the $1.4$-approximation algorithm for Graph TSP by Seb\H{o} and Vygen, we obtain a polynomial-time $(1.4+\epsilon)$-approximation algorithm for Graph Path TSP, improving on a recent $1.497$-approximation algorithm of Traub and Vygen.

We obtain our results through a variety of new techniques, including a novel way to set up a recursive dynamic program to guess significant parts of an optimal solution. 
At the core of our dynamic program we deal with instances of a new generalization of (Path) TSP which combines parity constraints 
with certain connectivity requirements.
This problem, which we call $\Phi$-TSP, has a constant-factor approximation algorithm
and can be reduced to TSP in certain cases when the dynamic program would not make sufficient progress.
\end{abstract}

\newpage

\section{Introduction}

The Traveling Salesman Problem (TSP) is one of the most fundamental and well-studied problems in Combinatorial Optimization with a multitude of applications.
The common denominator of the numerous variants of the problem is that a set of cities have to be visited on a shortest possible tour.
Its best-known variant, often just dubbed TSP, assumes that the distances between the cities are non-negative and symmetric, and the task is to find a tour beginning and ending in the same city.
For our purposes it will be useful to work with an undirected graph $G=(V,E)$ with edge lengths $\ell : E \rightarrow \mathbb{R}_{\ge 0}$. 
While it is often assumed that $G$ is complete and $\ell$ fulfills the triangle-inequality, we do not assume this, but allow the tour to 
visit cities more than once; this is easily seen (and well-known) to be equivalent.
So a tour is a closed walk in $G$ visiting all vertices.

One of the best-studied extensions of TSP is Path TSP, where in addition to $G$ and $\ell$,
a fixed start $s\in V$ and end $t\in V$ are given and the task is to find a shortest walk from $s$ to $t$ visiting all vertices.

TSP and its variants are well-known to be \APX-hard (see~\cite{papadimitriou_1993_traveling,karpinski_2015_new} and references therein) and they have been studied very extensively under the viewpoint of approximation algorithms. 
While TSP and Path TSP look quite similar, there are fundamental differences.
First, there is a classical $\sfrac{3}{2}$-approximation algorithm for TSP 
by Christofides~\cite{christofides_1976_worst-case} and Serdjukov~\cite{serdukov_1978_some}.
This algorithm can easily be adapted to Path TSP, but then has only approximation ratio $\sfrac{5}{3}$ 
as Hoogeveen~\cite{hoogeveen_1991_analysis} showed in the early 90's.
Second, the integrality gaps of the classical LP relaxations seem to be different.
For TSP it is widely believed to be $\sfrac{4}{3}$, while for Path TSP it is at least $\sfrac{3}{2}$.
In both cases there are well-known instances attaining these lower bounds on the integrality gaps, and these instances 
have unit lengths ($\ell(e) =1$ for all $e\in E$).
Therefore the unit-length special cases, Graph TSP and Graph Path TSP, have received considerable attention
\cite{an_2015_improving, oveisgharan_2011_randomized,momke_2016_removing,mucha_2014_approximation,sebo_2014_shorter,gao_2013_lp_based,traub_2018_beating}.
In these special cases, the integrality gaps are known to be different: it is at most $\sfrac{7}{5}$ for Graph TSP 
and exactly $\sfrac{3}{2}$ for Graph Path TSP as shown by Seb\H{o} and Vygen~\cite{sebo_2014_shorter}.

While Christofides' algorithm for TSP is still unbeaten after more than four decades, the approximation ratio for Path TSP has been improved.
The first improvement, about 20 years after Hoogeveen~\cite{hoogeveen_1991_analysis}, was obtained by An, Kleinberg, and Shmoys~\cite{an_2015_improving}, who devised an elegant $\sfrac{(1+\sqrt{5})}{2}\approx 1.618$-approximation algorithm.
A sequence of successive improvements~\cite{sebo_2013_eight-fifth,vygen_2016_reassembling,gottschalk_2016_better,sebo_2016_salesman,traub_2018_approaching}
culminated in Zenklusen's recent $\sfrac{3}{2}$-approximation algorithm \cite{zenklusen_2019_approximation}.
Hence, at the moment, the best known approximation ratios for TSP and Path TSP are the same.

For Graph TSP and Graph Path TSP the situation is different.
The best known approximation ratios are $\sfrac{7}{5}$ for Graph TSP \cite{sebo_2014_shorter} and $1.497$ for Graph Path TSP \cite{traub_2018_beating}.
Since the latter result achieves an approximation ratio better than the integrality gap $\sfrac{3}{2}$,
one might hope that Graph Path TSP is actually no harder than Graph TSP although the integrality gaps differ. 

These recent developments naturally lead to the following general question regarding the relation between the approximability of Path TSP and TSP, which we address in this paper:
\begin{center}
Is (Graph) Path TSP substantially harder to approximate than its well-known special case (Graph) TSP?  
\end{center}
The answer is no.
The main contribution of this paper is to show in a constructive way that Path TSP can be approximated equally well as TSP (up to an arbitrarily small error), by presenting a black-box reduction that transforms approximation algorithms for TSP into ones for Path TSP.

\subsection{Our results}

The main consequence of our reduction can be summarized as follows.
\begin{restatable}{theorem}{mainReduction}\label{thm:mainReduction}
Let $\mathcal{A}$ be an $\alpha$-approximation algorithm for TSP. Then, for any $\epsilon >0$, there is an $(\alpha+\epsilon)$-approximation algorithm for Path TSP that, for any instance $(G,\ell, s, t)$, calls $\mathcal{A}$ a strongly polynomial number of times on TSP instances defined on subgraphs of $(G,\ell)$, and performs further operations taking strongly polynomial time.
\end{restatable}
The following two statements are immediate consequences of the above theorem.
\begin{corollary}\label{cor:reductionPathTSPToTSP}
Let $\epsilon>0$ and $\alpha >1$. If there is a (strongly) polynomial-time $\alpha$-approximation algorithm for TSP, then there is a (strongly) polynomial-time $(\alpha+\epsilon)$-approximation algorithm for Path TSP.
\end{corollary}
\begin{corollary}\label{cor:reductionGraphTSP}
Let $\epsilon>0$ and $\alpha >1$. If there is a polynomial-time 
$\alpha$-approximation algorithm for Graph TSP, then there is a polynomial-time $(\alpha+\epsilon)$-approximation algorithm for Graph Path TSP.
\end{corollary}

Notice that since Graph (Path) TSP does not involve any large numbers in its input, the notions of polynomial-time and strongly polynomial-time algorithm are identical in this context.

The above statements create a strong link between the approximability of Path TSP and TSP, as well as its graph versions. More precisely, Theorem~\ref{thm:mainReduction} implies that such a link exists for any class of TSP instances
that is closed under taking instances on subgraphs of the original instance (without changing the edge lengths). 
In particular, any potential future progress on the approximability of (Graph) TSP will immediately carry over to (Graph) Path TSP.

Moreover, Corollary~\ref{cor:reductionGraphTSP} allows us to make significant progress on the currently best approximation factor of $1.497$ for Graph Path TSP~\cite{traub_2018_beating}, through a black-box reduction to the $1.4$-approximation algorithm for Graph TSP by Seb\H{o} and Vygen~\cite{sebo_2014_shorter}.
\begin{corollary}
For any $\epsilon>0$, there is a polynomial-time $(1.4+\epsilon)$-approximation algorithm for Graph Path TSP.
\end{corollary}

Our reduction technique is quite versatile. 
In particular, it applies to a pretty general problem class (the $\Phi$-tour problem with interfaces of bounded size; see 
Definition~\ref{def:phi_tour} and Theorem~\ref{thm:phiTourMain}).
This includes the $T$-tour problem for bounded $|T|$ (see~\cite{sebo_2014_shorter,cheriyan_2015_approximating_2,sebo_2013_eight-fifth} for a definition) and certain uncapacitated vehicle routing problems such as the one with a fixed number of depots studied in \cite{Xu_2015_3_2}.

\subsection{Organization of the paper}

After some brief preliminaries in Section~\ref{sec:preliminaries} to fix basic terminology and notation, we provide an overview of our approach in Section~\ref{sec:overview}. Here, we first focus on some key aspects of our approach, which is based on a new way to employ dynamic programming by using a well-chosen auxiliary problem, which we call $\Phi$-TSP. Moreover, we break down the problem of finding a short solution to $\Phi$-TSP into two cases. 
Combining the two cases, applying the same algorithm recursively, and using a constant-factor approximation algorithm for $\Phi$-TSP on the final recursion level
will imply our main reduction result, Theorem~\ref{thm:mainReduction}.

For one case, 
we show in Section~\ref{sec:shortTJoin} how to reduce the problem to TSP.
For the other case,
we show in Section~\ref{sec:dynProg} how to guess a constant fraction of an optimum solution via dynamic programming.
The detailed proof of Theorem~\ref{thm:mainReduction} is in Section~\ref{sec:MainThm}.
Finally, Section~\ref{sec:4-Approx} contains a $4$-approximation algorithm for $\Phi$-TSP, 
which is followed by concluding remarks in Section~\ref{sec:conclusions}.

\section{Preliminaries}\label{sec:preliminaries}

A weighted graph is a tuple $(V,E,\ell)$, where $V$ is the vertex set, $E$ is the edge set, which we assume w.l.o.g.\ not to have loops or parallel edges, and $\ell : E\to \mathbb{R}_{\geq 0}$ denotes the edge lengths.
We only consider undirected graphs with non-negative edge lengths and do not always state this explicitly.

We often deal with multi-sets of edges.
Although $E$ does not contain parallel edges, when we write $F\subseteq E$, we mean a multi-set $F$ that can contain several copies of the same edge.
We use the operator $\cupp$ to designate the multi-union.
For a vertex set $W\subseteq V$, we denote by $\delta(W)\subseteq E$ all edges with exactly one endpoint in $W$, and, for $v \in V$, we use $\delta(v)$ as a shorthand for $\delta(\{v\})$.
For a multi-set $F\subseteq E$, we define
\begin{equation*}
\odd(F)\coloneqq \{v : v\text{ is a vertex with } |\delta(v)\cap F| \text{ is odd}\}\enspace.
\end{equation*}
For a vertex set $T\subseteq V$, a $T$-join is a multi-set of edges $F\subseteq E$ with $\odd(F) = T$.

For a set $I\subseteq V$ and a graph $G$, we denote by $G/I$ the graph obtained from $G$ by contracting the vertex set $I$. If $I$ is empty, we define $G/\emptyset:=G$.
For a vertex set $W$ and an edge set $F$, we define $F[W]\coloneqq \{e\in F: \text{both endpoints of $e$ are in $W$}\}$.
Moreover, $G[W]$ denotes the induced subgraph with vertex set $W$ and edge set $E[W]$.

Instead of describing tours as walks in $G$, it is convenient to consider them as multi-edge sets.
Then a solution to (Graph) TSP is a multi-edge set $F$ such that $(V,F)$ is connected and $\odd(F) = \emptyset$.
A solution to (Graph) Path TSP (with $s\ne t$) is a multi-edge set $F$ such that $(V,F)$ is connected and $\odd(F) = \{s, t\}$.
From such multi-edge sets---which we also call \emph{tours} or \emph{$s$-$t$ tours}, respectively---we can easily recover walks by Euler's theorem.

We often use $\OPT$ as an arbitrary (but fixed) optimal solution for the problem in question. 
Finally, when using the notion of \emph{approximation algorithm} we will not assume that the algorithm is polynomial time, but state it explicitly if this is the case.

In the interest of clarity and simplicity of the presentation, we did not try to optimize the running times of our procedures. Consequently, we often opt for weaker constants that are easier to obtain.

\section{Overview of approach}\label{sec:overview}

A key novelty of our approach is a new way to set up a dynamic program to successively strengthen a basic algorithm by combining it with a stronger algorithm for TSP. Every time we apply our dynamic program to obtain a stronger algorithm, we end up with a more difficult problem, slowly approaching problem settings for which it is very challenging to find strong approximation algorithms. However, as we show, by guessing a well-chosen set of edges through the dynamic program, we can limit the recursion depth by a constant, which allows us to stay in a regime where our approach runs efficiently.

To introduce our approach, we start with a brief discussion of a much more basic dynamic programming idea that has previously been used in related settings.
We explain the challenges this procedure faces when trying to extend it for our purposes, and 
outline how we overcome the barriers encountered by existing methods.

\subsection{Key challenges and high-level approach}\label{sec:high_level}
 
Assume we are given an $\alpha$-approximation algorithm $\mathcal{A}$ for TSP.
Then finding a short Path TSP solution using $\mathcal{A}$ as an oracle would be easy if the distance $d(s,t)$ between the start $s$ and the end $t$ was short compared to $\ell(\OPT)$, i.e., the length of a shortest Path TSP solution $\OPT$. Indeed, in this case the length of a shortest TSP tour $\OPT_{\TSP}$ and a shortest $s$-$t$ tour $\OPT$ do not differ by much because any solution of one problem can be converted to a solution of the other one by adding a shortest $s$-$t$ path $P$. More precisely, $\OPT_{\TSP} \cupp P$ is an $s$-$t$ tour and $\OPT \cupp P$ is a TSP tour.
Hence, one can simply compute an $\alpha$-approximate TSP tour $F$ and a shortest $s$-$t$ path $P$ and return $F \cupp P$. 

Consequently, a canonical plan would be to try to transform the Path TSP instance to another one with small $s$-$t$ distance.  
It turns out that if the distance between $s$ and $t$ is very large, then such a reduction is indeed possible by using a technique based on dynamic programming that goes back to Blum, Chawla, Karger, Lane, Meyerson, and Minkoff~\cite{blum_2007_approximation}, who studied variants of the Orienteering Problem. Their approach was later extended by Traub and Vygen~\cite{traub_2018_beating} in the context of Graph Path TSP.
This approach allows for reducing to Path TSP instances where the distance between $s$ and $t$ is at most $(\sfrac{1}{3}+\epsilon)\cdot \ell(\OPT)$, for some arbitrarily small constant $\epsilon > 0$ (see~\cite{traub_2018_beating}).

However, this technique faces significant barriers when aiming at a reduction to smaller $s$-$t$ distances. 
Thus our approach follows a different path.
Nevertheless, it is on a high level inspired by the dynamic program in~\cite{blum_2007_approximation} and later variations and extensions thereof~\cite{traub_2018_approaching,zenklusen_2019_approximation,traub_2018_beating,nagele_2019_new}. 
We therefore start with a brief discussion of this prior technique in the context of Path TSP as used in~\cite{traub_2018_beating}, which will be helpful for the understanding of our approach.

\medskip

For simplicity of exposition, consider a Graph Path TSP instance and assume that $d(s,t) \geq (\sfrac{1}{3}+\epsilon)\cdot |\OPT|$ 
for some constant $\epsilon > 0$.
The idea is to study the structure of edges of $\OPT$ in the $d(s,t)$ many $s$-$t$ cuts $\delta(L_0), \ldots, \delta(L_{d(s,t)-1})$ defined by
\begin{equation*}
L_{i} \coloneqq \left\{v\in V : d(s,v) \leq i \right\} \qquad \forall i\in {0,\ldots, d(s,t)-1}\enspace.
\end{equation*}
The key observation is that a constant fraction of these $s$-$t$ cuts will only contain a single edge of $\OPT$, and, hence, one can try to ``guess'' these edges through a dynamic program.
Indeed, every edge can be in at most one of the cuts $\delta(L_0), \delta(L_1),\ldots,\delta(L_{d(s,t)-1})$.
Hence, the average number of $\OPT$-edges in a cut $\delta(L_i)$ can be no higher than $\sfrac{|\OPT|}{d(s,t)}$. 
Using that every $s$-$t$ cut must have an odd intersection with $\OPT$, because $\OPT$ is an $s$-$t$ tour, this implies that a constant fraction of the cuts contains only one edge of $\OPT$.
For brevity, we call a cut $\delta(L_i)$ with $|\delta(L_i) \cap \OPT|=1$ a \emph{$1$-cut}.
Assume we knew all edges of $\OPT$ contained in $1$-cuts. Then the problem decomposes into smaller Path TSP instances. 
See Figure~\ref{fig:simpleDPGraphPathTSP} for an illustration.

\begin{figure}[h!]

\begin{center}
\begin{tikzpicture}[scale=1]

\small

\pgfdeclarelayer{bg}
\pgfdeclarelayer{fg}
\pgfsetlayers{bg,main,fg}

\def\ch{4.5} %
\def\dx{1.5} %

\tikzset{oneEdge/.style={line width=3pt,blue}}

\begin{scope}[every node/.style={fill=black,circle,inner sep=0em,minimum size=4pt},yshift=5.9cm]
\node  (1) at (0.76,-5.90) {};
\node  (3) at (2.20,-6.46) {};
\node  (4) at (2.3,-7.56) {};
\node  (5) at (3.8,-6.56) {};
\node  (7) at (3.72,-4.60) {};
\node  (8) at (5.18,-4.24) {};
\node  (9) at (6.16,-4.08) {};

\node (10) at (7.8,-4.0) {};
\node (10b) at (7.15,-4.5) {};
\node (10c) at (7.65, -5.1) {};

\node (11) at (6.42,-5.24) {};
\node (12) at (5.22,-6.16) {};
\node (13) at (6.62,-6.64) {};
\node (14) at (8.2,-5.60) {};
\node (15) at (9.6,-5.94) {};
\node (16) at (11.0,-7.28) {};
\node (17) at (12.60,-7.00) {};
\node (18) at (11.36,-5.88) {};
\node (19) at (11.46,-4.46) {};
\node (20) at (12.68,-4.70) {};
\node (21) at (14.2,-5.90) {};
\end{scope}

\begin{scope}
\node[above=1pt] at (1) {$s=u_0$};
\node[below=1pt] at (21) {$t=v_8$};

\node[above=2pt] at (3) {$v_0$};
\node[above=2pt] at (7) {$u_2$};
\node[above=2pt] at (8) {$v_2$};
\node[above=1pt] at (14) {$u_5$};
\node[above=1.5pt,xshift=2mm] at (15) {$v_5 = u_6$};
\node[below=1pt] at (16) {$v_6$};
\node[above=1pt] at (20) {$u_8$};

\end{scope}

\def \bs{20pt}

\begin{scope}[thick]
\draw[oneEdge]  (1) --  (3);
\draw  (3) -- (4) -- (5) -- (3) -- (7);
\draw[oneEdge]  (7) --  (8);
\draw  (8) --  (9) -- (10b) -- (10c) -- (11) -- (12) -- (13) -- (10c) -- (14);
\draw (10b) to[bend right=\bs] (10);
\draw (10b) to[bend left=\bs] (10);

\draw[oneEdge] (14) -- (15);
\draw[oneEdge] (15) -- (16);
\draw (16) -- (17) -- (18) -- (19) -- (20);
\draw[oneEdge] (20) -- (21);
\end{scope}

\begin{scope}[dashed,bend right=12]
\newcommand\cut[2][]{
\draw (#2,0.1cm-\ch cm /2) node[below=2pt, left=0pt] {#1} to ++ (0,\ch);
}

\begin{scope}
\foreach \i in {1,3,4,7} {
\cut[$L_\i$]{1.2 + \i*\dx}
}
\end{scope}

\begin{scope}[blue,thick]
\foreach \i in {0,2,5,6,8} {
\cut[$L_\i$]{1.2 + \i*\dx}
}
\end{scope}

\end{scope}

\end{tikzpicture}

 \end{center}
\caption{Illustration of an optimal solution $\OPT$ of a Graph Path TSP instance with $9=d(s,t) > \sfrac{1}{3}\cdot |\OPT| = \sfrac{23}{3}$. 
Thick blue edges are OPT-edges that are the only one in some cut $\delta(L_i)$.
Knowing these edges, the problem breaks into six smaller instances: 
a trivial $s$-$u_0$ Path TSP problem in $G[L_0]$, one $v_0$-$u_2$ Path TSP problem in $G[L_2\setminus L_0]$ and so on, with the last instance being a trivial $v_8$-$t$ Path TSP problem in $G[V\setminus L_8]$. \\
}
\label{fig:simpleDPGraphPathTSP}

\end{figure}
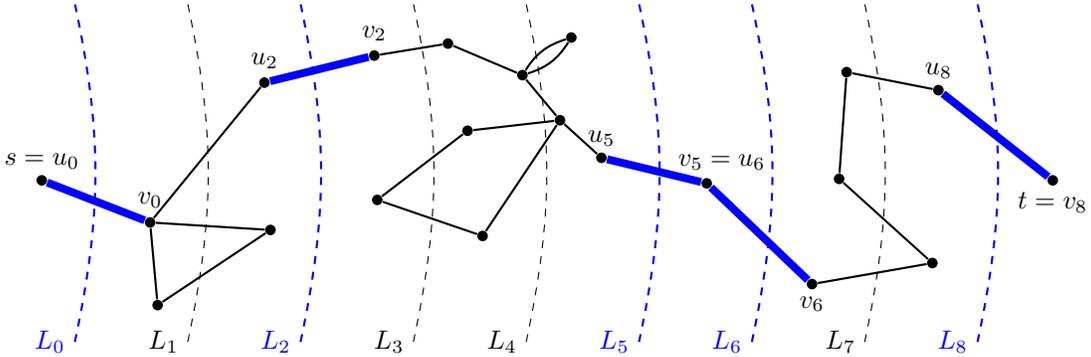
 
Of course, the $\OPT$-edges in $1$-cuts are not known upfront, and hence, the problem cannot be decomposed so easily. However, one can use a dynamic program to guess the $1$-cuts from left to right, i.e., from $\delta(L_0)$ to $\delta(L_{d(s,t)-1})$, together with the $\OPT$-edge in each of them. 
Notice that the sub-instances may not have a short start-to-end distance 
(e.g. $d(v_6,u_8)$ in Figure~\ref{fig:simpleDPGraphPathTSP} may be substantially larger than $\sfrac{1}{3}$ times an optimum
$v_6$-$v_8$-tour in $G[L_8\setminus L_6]$).
As shown in~\cite{traub_2018_beating}, this issue can be addressed by applying the dynamic program recursively to the sub-instances. A key observation in~\cite{traub_2018_beating} is that a constant recursion depth is enough to ensure that the total cost of the remaining sub-instances
becomes negligible compared to the edges guessed through the recursive dynamic program.

\medskip

Notice that to apply this dynamic programming idea, one crucially needs $d(s,t) \geq (\sfrac{1}{3}+\epsilon) \cdot |\OPT|$ for some constant 
$\epsilon > 0$.
Indeed, otherwise none of the cuts $\delta(L_i)$ may be a $1$-cut, and no decomposition into smaller Path TSP instances as above is possible. 
This is the reason why this techniques has only been applied to reduce the start-to-end distance to about a third of $\OPT$.

If we could guess not only $1$-cuts, but also cuts with a larger constant number of $\OPT$-edges, say up to $5$, then we could handle instances with an $s$-$t$ distance below $\sfrac{1}{3}\cdot \ell(\OPT)$. (This idea is inspired by a recent dynamic programming approach in~\cite{nagele_2019_new} in the context of chain-constrained spanning trees.)
Our approach aims at realizing this high-level plan. However, this ostensibly simple algorithmic idea comes with several significant technical hurdles. Most importantly, if we guess more than one edge, the resulting sub-problems are not Path TSP problems anymore. 
More precisely, if we guess 5 edges in each of two consecutive $5$-cuts, then we have up to 10 \emph{interface vertices}, i.e., endpoints of guessed edges. 
See Figure~\ref{fig:guess3} for an example.

\begin{figure}[h!]

\begin{center}

\begin{tikzpicture}[scale=0.6]

\small

\pgfdeclarelayer{bg}
\pgfdeclarelayer{fg}
\pgfsetlayers{bg,main,fg}

\tikzset{oneEdge/.style={line width=3pt,blue}}

\def\ch{9.5} %

\begin{scope}[dashed,bend right=12]
\newcommand\cut[2][]{
\draw (#2,0.5) node[below=2pt, left=0pt] {#1} to ++ (0,\ch);
}

\begin{scope}
\cut[$L_2$]{8.5}

\end{scope}

3-cuts
\begin{scope}[blue,thick]
\foreach \i/\j in {1/0,4/1,14/3} {
\cut[$L_\j$]{0 + \i}
}
\end{scope}
\end{scope}

\def\rad{4}
\def\cd{0.5}
\def \bs{20pt}

\tikzset{nsT/.style={
draw=red,thick,rectangle,inner sep=0em,minimum size=5pt,fill=white
}}

\tikzset{nsI/.style={
fill=white,draw=red, thick,circle,inner sep=0em,minimum size=5pt
}}

\tikzset{ns/.style={
fill=black,circle,inner sep=0em,minimum size=4pt}
}

\begin{scope}[every node/.style={nsI}, shift={(9.5,5.5)}]
\node[nsT] (v1) at (0:\rad) {};
\node[nsT] (v2) at (-45:\rad) {};

\node[nsT] (v3) at (-90:\rad) {};
\node[nsT] (v4) at (-135:\rad) {};

\node (v5) at (180:\rad) {};

\node (v6) at (135:\rad) {};
\node[nsT] (v7) at (90:\rad) {};
\node[nsT] (v8) at (45:\rad) {};
\end{scope}

\begin{scope}[every node/.style={ns}]
 \node (n1) at (7,4) {};
 \node (n2) at (7,5.5) {};
 \node (n3) at (8,6.5) {};
 \node (n4) at (10,6) {};
 
 \node (n5) at (8,4) {};
 \node (n6) at (11,5) {};

 \node (n8) at (10,7) {};
\end{scope}

\begin{scope}[thick]
\draw (v3) -- (n1);
\draw (v4) -- (n1);
\draw (n1) to[bend right=\bs] (n2);
\draw (n1) to[bend left=\bs] (n2);
\draw (n2) -- (n3) -- (n4) -- (n2);

\draw (v2) -- (n5) -- (n6) -- (v1);

\draw (v7) -- (v6) -- (n8) -- (v8);
\end{scope}

\begin{scope}[every path/.style={fill=black!12,draw=black!12}, shift={(9.5,5.5)}]
\begin{pgfonlayer}{bg}
\newcommand\Cset[2]{
\draw (#1:\rad+\cd) arc (#1:#2:\rad+\cd)
  arc (#2:#2+180:\cd)
  arc (#2:#1:\rad-\cd)
  arc (#1+180:#1+360:\cd)
;
}

\Cset{-45}{0}
\Cset{-135}{-90}
\Cset{180}{180}
\Cset{45}{135}
\end{pgfonlayer}
\end{scope}

\begin{scope}[every node/.style={ns}]
 \node (w1) at (3,9) {};
 \node (w2) at (3,8) {};
 
 \node (w3) at (3.5,5.5) {};
 \node (w4) at (2.5, 6.5) {};
 \node (w5) at (2.5, 4.5) {};
 
 \node (w6) at (3,3) {};
 
 \node (s) at (0,5) {};
 
 \node (w7) at (15.5,8.5) {};
 \node (t) at (16,2) {}; 

\end{scope}
 \node[left=2pt] at (s) {$s$};
 
\node[right=2pt] at (t) {$t$};

\begin{scope}[every node/.style={ns}]
  \node (m1) at (2,8.5) {};
\end{scope}

\begin{scope}[thick]
 \draw (w1) -- (m1) -- (w2);
 \draw (w4) -- (w3) -- (w5);
 \draw (w7) -- (t);
\end{scope}

\begin{scope}[oneEdge]
 \draw (w1) -- (v6) -- (w2);
 \draw (v5) to[bend right=\bs] (w3);
 \draw (v5) to[bend left=\bs] (w3);
 \draw (w4) -- (s) -- (w5);
 \draw (s) -- (w6) -- (v4);
 \draw (v3) -- (t) -- (v2);
 \draw (v7) -- (w7) -- (v8);
 \draw (w7) -- (v1);
\end{scope}

\end{tikzpicture}

 \end{center}
\caption{Illustration for guessing edges in cuts $\delta(L_i)$ that contain at most five $\OPT$-edges. 
In this example, $\delta(L_0), \delta(L_1)$, and $\delta(L_3)$ are these cuts and the edges crossing them are highlighted as thick blue edges.
We now define a sub-problem in $G[L_3\setminus L_1]$.
We call the endpoints of the thick blue edges \emph{interface vertices} (red).
The gray sets show the connectivity requirements on the interfaces vertices in this sub-problem: 
in a solution of the sub-problem, vertices in the same gray set must be in the same connected component.
Interface vertices shown as squares must have odd degree; all other vertices must have even degree.}
\label{fig:guess3}

\end{figure}
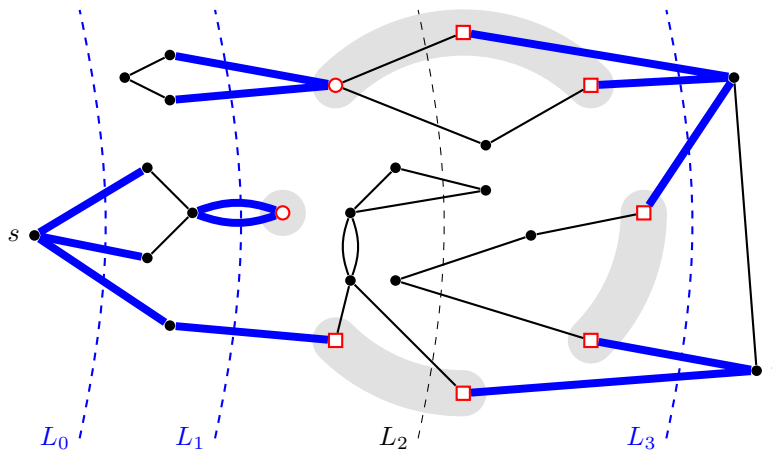

An optimum $s$-$t$ tour is not necessarily connected inside the vertex set of a sub-problem
but every connected component must contain at least one interface vertex.
Moreover, $\OPT$ needs to connect some of the interface vertices to each other.
This induces connectivity constraints for the sub-problem, shown as gray sets in Figure~\ref{fig:guess3}.
They can also be guessed since the number of interface vertices is constant.
Note, however, that we cannot guess the entire connected components, as there are exponentially many options.

Clearly, these sub-problems become significantly more difficult than the original Path TSP problem. Moreover, if we try to apply such a procedure recursively, then the sub-problems can become more complex with each recursion step, because of an increase in the number of interface vertices per sub-problem. 
Another important issue in a recursive application to our more complex sub-problem
is to identify good cuts in which we should guess edges of $\OPT$. 
Our cuts will result from the dual of a $T$-join problem. 
They will no longer form a chain, but their laminar structure still allows for a dynamic programming approach.

Moreover, it is not obvious how to reduce the problem to TSP in the case when we cannot guess edges by dynamic programming,
and this will involve a careful guessing of further edges of $\OPT$.

We will now describe our approach in detail.
We start by defining a new problem class around which our method is centered, and which we call $\Phi$-TSP.

\subsection{\boldmath $\Phi$-TSP}

As described above, when guessing edges, the endpoints of those edges play a special role in terms of how we have to connect things. 
We capture this through the notion of an interface. 
We define this notion for a general graph $G$ below and will typically use it for subgraphs of the instance we are interested in.
\begin{definition}[interface]
An interface $\Phi$ of a graph $G=(V,E)$ is a triple $\Phi=(I,T,\mathcal{C})$ with
\begin{enumerate}
\item $T\subseteq I\subseteq V$, where $|T|$ is even, and 
\item $\mathcal{C}\subseteq 2^I$ is a partition of $I$.
\end{enumerate}
For an interface $\Phi$ of $G$, we denote by $(I_\Phi, T_\Phi, \mathcal{C}_\Phi)$ its corresponding triple and call $|I_\Phi|$ its size.
\end{definition}

For a given interface, we are interested in finding what we call \emph{$\Phi$-tours}, which are defined as follows.
\begin{definition}[$\Phi$-tour]\label{def:phi_tour}
Let $G=(V,E)$ be a graph.
Let $\Phi=(I,T,\mathcal{C})$ be an interface of $G$. A $\Phi$-tour in $G$ is a multi-set $F\subseteq E$ with
\begin{enumerate}
\item\label{item:PhiTourIsTJoin} $T = \odd(F)$, i.e., $F$ is a $T$-join,
\item\label{item:PhiTourConnectToInterface} $(V,F)/ I$ is connected, and
\item\label{item:PhiTourConnectivity} for any $C\in \mathcal{C}$, the vertices in $C$ lie in the same connected component of $(V,F)$.
\end{enumerate}
\end{definition}
Figure~\ref{fig:phiTour} exemplifies the notation of an interface $\Phi$ and a $\Phi$-tour.

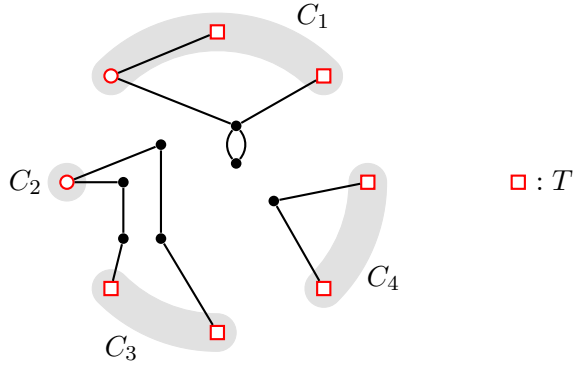
\begin{figure}[h!]

\begin{center}
\hspace{2.5cm}%
\begin{tikzpicture}[scale =0.5]

\pgfdeclarelayer{fg}
\pgfdeclarelayer{bg}
\pgfsetlayers{bg,main,fg}

\def\rad{4}
\def\cd{0.5}
\def \bs{40pt}

\tikzset{nsT/.style={
draw=red,thick,rectangle,inner sep=0em,minimum size=5pt,fill=white
}}

\tikzset{nsI/.style={
fill=white, draw=red,thick,circle,inner sep=0em,minimum size=5pt
}}

\tikzset{ns/.style={
fill=black,circle,inner sep=0em,minimum size=4pt}
}

\begin{scope}[every node/.style={nsI}]
\node[nsT] (v1) at (0:\rad) {};
\node[nsT] (v2) at (-45:\rad) {};

\node[nsT] (v3) at (-90:\rad) {};
\node[nsT] (v4) at (-135:\rad) {};

\node (v5) at (180:\rad) {};

\node (v6) at (135:\rad) {};
\node[nsT] (v7) at (90:\rad) {};
\node[nsT] (v8) at (45:\rad) {};
\end{scope}

\begin{scope}
\node[nsT] (TLab) at (8,0) {};
\node[right] at ($(TLab)+(0.2,0.06)$) {$:T$};
\end{scope}

\begin{scope}[every path/.style={fill=black!12,draw=black!12}]
\begin{pgfonlayer}{bg}
\newcommand\Cset[2]{
\draw (#1:\rad+\cd) arc (#1:#2:\rad+\cd)
  arc (#2:#2+180:\cd)
  arc (#2:#1:\rad-\cd)
  arc (#1+180:#1+360:\cd)
;
}

\Cset{-45}{0}
\Cset{-135}{-90}
\Cset{180}{180}
\Cset{45}{135}
\end{pgfonlayer}
\end{scope}

\begin{scope}
\def\td{0.6}
\node at (60:\rad+\cd+\td) {$C_1$};
\node at (180:\rad+\cd+\td) {$C_2$};
\node at (-120:\rad+\cd+\td) {$C_3$};
\node at (-30:\rad+\cd+\td) {$C_4$};
\end{scope}

\begin{scope}[every node/.style={ns}, shift={(-9.5,-5.5)}]
 \node (n1) at (7,4) {};
 \node (n2) at (7,5.5) {};
 \node (n3) at (8,6.5) {};
 \node (n4) at (10,6) {};
 
 \node (n5) at (8,4) {};
 \node (n6) at (11,5) {};

 \node (n8) at (10,7) {};
\end{scope}

\begin{scope}[thick, shift={(-9.5,-5.5)}]
\draw (v3) -- (n5) -- (n3);
\draw (v4) -- (n1);
\draw (n1) -- (n2);
\draw (n8) to[bend right=\bs] (n4);
\draw (n8) to[bend left=\bs] (n4);
\draw (n2) -- (v5) -- (n3) ;

\draw (v2) -- (n6) -- (v1);

\draw (v7) -- (v6) -- (n8) -- (v8);
\end{scope}

\end{tikzpicture}
 \end{center}
\caption{Example of an interface $\Phi=(I,T,\mathcal{C})$ and a $\Phi$-tour. The partition $\mathcal{C}=\{C_1,C_2,C_3,C_4\}$ of $I$ is highlighted as gray sets. 
Hence, $I = \cup_{i=1}^4 C_i$. Moreover, the vertices in $T$ are drawn as red rectangles.
This defines a $\Phi$-TSP instance that results as a sub-problem in Figure~\ref{fig:guess3}. The shown edges $F$ are a $\Phi$-tour.
}
\label{fig:phiTour}

\end{figure}
The problem we focus on in the following, which we call $\Phi$-TSP, seeks to find a shortest $\Phi$-tour.
\begin{definition}[$\Phi$-TSP]
Given a weighted graph $G=(V,E,\ell)$ and an interface $\Phi$ of $G$, compute a shortest $\Phi$-tour in $G$ or decide that none exists.
In short,
\begin{equation}\label{eq:PhiTSP}
\min\left\{\ell(F) : F \text{ is a $\Phi$-tour in $G$}  \right\}\enspace. \tag{$\Phi$-TSP}
\end{equation}
\end{definition}

Note that for any distinct $s,t\in V$, by choosing the interface $\Phi=(I,T,\mathcal{C})$ with $I=T=\{s,t\}$ and $\mathcal\{C\}=\{\{s,t\}\}$, we have that $\Phi$-tours correspond to solutions to $s$-$t$ Path TSP. Analogously, for larger sets $T$, one captures the $T$-tour problem (see~\cite{sebo_2014_shorter,cheriyan_2015_approximating_2,sebo_2013_eight-fifth}).
Another special case is the uncapacitated vehicle routing problem with a fixed number of depots, for which Xu and Rodrigues \cite{Xu_2015_3_2}
gave a $\sfrac{3}{2}$-approximation.
Here, $I$ is the set of depots, $T=\emptyset$, and $\mathcal{C}$ is the partition into singletons.

Depending on the structure of the graph $G$ and the interface $\Phi$, it may be that no $\Phi$-tour exists. 
We call an interface $\Phi$ of $G$ \emph{feasible} if $G$ admits a $\Phi$-tour.
The existence of a $\Phi$-tour admits the following easy characterization, which can be checked in linear time.
\begin{lemma}\label{lem:existencePhiTour}
Let $G=(V,E,\ell)$ be a weighted graph.
Let $\Phi=(I,T,\mathcal{C})$ be an interface of $G$. Then $G$ admits a $\Phi$-tour if and only if all of the following conditions hold.
\begin{enumerate}
\item\label{item:PhiTourExistsJoin} Each connected component of $G$ contains an even number of vertices in $T$,
\item\label{item:PhiTourExistsX} $G/I$ is connected, and
\item\label{item:PhiTourExistsI} for every $C\in \mathcal{C}$, the vertices in $C$ lie in the same connected component of $G$.
\end{enumerate}
\end{lemma}
\begin{proof}
The three mentioned conditions are clearly necessary for $G$ to admit a $\Phi$-tour. Moreover, if they are satisfied then, due to~\ref{item:PhiTourExistsJoin}, there exists a $T$-join $J\subseteq E$, and points~\ref{item:PhiTourExistsX} and~\ref{item:PhiTourExistsI} guarantee that $E \cupp (E\setminus J)$ is a $\Phi$-tour in $G$.
\end{proof}

It is crucial for our approach to start with a polynomial-time constant-factor approximation algorithm, which we will successively strengthen as discussed in the following.

A $7$-approximation algorithm for $\Phi$-TSP can be obtained easily as follows.
Compute
a minimum cost edge set $F_1$ satisfying \ref{item:PhiTourExistsJoin} ($T$-join), 
a minimum cost edge set $F_2$ satisfying \ref{item:PhiTourExistsX} (spanning tree in $G/I$),
and a 2-approximation $F_3$ of a minimum cost edge set satisfying \ref{item:PhiTourExistsI} (Steiner forest).
Then the disjoint union $F_1 \cupp F_2 \cupp F_2 \cupp F_3 \cupp F_3$ is a $7$-approximation.

With a little more care we can obtain a $4$-approximation algorithm, using Jain's iterative rounding framework~\cite{jain_2001_factor}:

\begin{restatable}{theorem}{fourApprox}\label{thm:fourApprox}
$\Phi$-TSP admits a strongly polynomial $4$-approximation algorithm.
\end{restatable}

We defer the proof to Section~\ref{sec:4-Approx}.
In the rest of this paper, we will derive a strongly polynomial $(\alpha+\epsilon)$-approximation algorithm for
$\Phi$-TSP instances with bounded interface size, where $\alpha$ is the approximation guarantee for TSP; see Theorem~\ref{thm:phiTourMain}.

\subsection{Iterative improvement of basic algorithm}\label{sec:boosting_thm}

For a TSP algorithm $\Ascr$, we denote for every weighted graph $G$ by $f_{\Ascr}(G)$ the maximum 
runtime of algorithm $\Ascr$ on any subgraph of $G$.
Similarly, for a $\Phi$-TSP algorithm $\Bscr$, we denote for every weighted graph $G$ and any $k\in\mathbb{R}_{\ge 0}$ by 
$f_{\Bscr}(G,k)$ the maximum runtime of algorithm $\Bscr$ on any instance $(G',\Phi)$, where 
$G'$ is a subgraph of $G$ and $|I_{\Phi}| \le k$.

Our plan is to start with the $4$-approximation algorithm for $\Phi$-TSP guaranteed by Theorem~\ref{thm:fourApprox}, and successively improve it through a TSP algorithm with an approximation guarantee $\alpha$.
The following Boosting Theorem is the main technical result towards this goal and quantifies the improvement in terms of approximation factor that we are able to obtain in one improvement step.

\begin{theorem}[Boosting Theorem]\label{thm:iterImprovement}
Let $\alpha, \beta > 1$. Suppose we are given:
\begin{enumerate}[label=(\alph*)]
\item\label{item:iterTSP} an $\alpha$-approximation algorithm $\mathcal{A}$ for TSP, and
\item\label{item:iterPhiTSP} a $\beta$-approximation algorithm $\mathcal{B}$ for $\Phi$-TSP.
\end{enumerate}

Then there is an algorithm that, for any $\epsilon \in (0,1]$, any weighted graph $G=(V,E,\ell)$, and any feasible interface $\Phi = (I,T,\mathcal{C})$ of $G$, returns a $\Phi$-tour $F$ in $G$ of length
\begin{equation}\label{eq:boostingGuarantee}
\ell(F) \leq \max\left\{ (1+\epsilon)\alpha,\ \beta - \frac{\epsilon}{8}(\beta-1) \right\}\cdot \ell(\OPT)
\end{equation}
in time $|V|^{O\left(\frac{|I|}{\epsilon}\right)}
\cdot \left(f_{\mathcal{A}}(G) + f_{\mathcal{B}}\left(G, \frac{9|I|}{\epsilon}\right)\right)$, where $\OPT$ is a shortest $\Phi$-tour in $G$. 
In particular, the algorithm makes calls to $\mathcal{B}$ only on instances with interfaces of size bounded by $\frac{9 |I|}{\epsilon}$. 
\end{theorem}

To prove Theorem~\ref{thm:mainReduction}, we start with $\beta = 4$ (Theorem~\ref{thm:fourApprox}) and apply Theorem~\ref{thm:iterImprovement}
repeatedly, but only a constant number of times.
The approximation guarantee $\beta$ decreases until it reaches $(1+\epsilon)\alpha$.
All interfaces will have constant size. We defer the details to Section~\ref{sec:MainThm}.

\subsection{Proof outline of Boosting Theorem (Theorem~\ref{thm:iterImprovement})}\label{sec:proof_of_boosting}

Theorem~\ref{thm:iterImprovement} is obtained by designing two algorithms to obtain a $\Phi$-tour and then returning the better of the $\Phi$-tours computed by these algorithms. 
Each of the following two theorems summarizes the guarantee we obtain with one of the two algorithms.
After that, Algorithm~\ref{algo:iterImprovement}, described below, combines these two sub-procedures to obtain an algorithm that implies Theorem~\ref{thm:iterImprovement}.

The following theorem yields a short $\Phi$-tour if the length of a minimum $T_{\Phi}$-join is small. 
\begin{theorem}\label{thm:algShortTJoin}
Let $\alpha >1$.
Assume we are given an $\alpha$-approximation algorithm $\mathcal{A}$ for TSP. 
Then, for any $\delta >0$, any weighted graph $G=(V,E,\ell)$, and any feasible interface $\Phi=(I,T,\mathcal{C})$ of $G$, one can determine a $\Phi$-tour $F$ in $G$ with 
\begin{equation*}
\ell(F)\leq (1+\delta) \cdot \alpha \cdot \ell(\OPT) + (\alpha+1)\cdot \ell(J)%
\end{equation*}
in time $|V|^{O\left(\frac{|I|}{\delta}\right)} \cdot f_{\mathcal{A}}(G)$,
where $J$ is a shortest $T$-join in $G$ and $\OPT$ is a shortest $\Phi$-tour in $G$.
\end{theorem}

We will give the proof in Section~\ref{sec:shortTJoin}.
The next theorem, proven in Section~\ref{sec:dynProg}, states that we also obtain a short $\Phi$-tour if the length of a minimum $T$-join is large. 

\begin{theorem}\label{thm:algLongTJoin}
Let $\beta > 1$.
Assume we are given a $\beta$-approximation algorithm $\mathcal{B}$ for $\Phi$-TSP.
Then, for any $\delta >0$, any weighted graph $G=(V,E,\ell)$, and any feasible interface $\Phi=(I,T,\mathcal{C})$ of $G$, 
one can determine a $\Phi$-tour $F$ in $G$ with
\begin{equation*}
\ell(F) \leq \left(\big. \beta+ \delta\cdot(\beta-1)\right)\cdot \ell(\OPT) - (\beta-1)\cdot \ell(J)%
\end{equation*}
in time
 $|V|^{O\left(|I| +\frac{|T|}{\delta}\right)} \cdot f_{\mathcal{B}}\left(G, |I|+\frac{|T|}{\delta}\right)$, 
 where $J$ is a shortest $T$-join in $G$ and $\OPT$ is a shortest $\Phi$-tour in $G$.
\end{theorem}

\begin{algorithm2e}[H]
\begin{enumerate}[label=\arabic*.]
\item Run algorithm guaranteed by Theorem~\ref{thm:algShortTJoin} with $\delta = \sfrac{\epsilon}{2}$ to obtain $\Phi$-tour $F_1$.

\item Run algorithm guaranteed by Theorem~\ref{thm:algLongTJoin} with $\delta = \sfrac{\epsilon}{8}$ to obtain $\Phi$-tour $F_2$.

\item Return the shorter $\Phi$-tour among $F_1$ and $F_2$.

\end{enumerate}

\caption{Approximation algorithm for $\Phi$-TSP to prove Theorem~\ref{thm:iterImprovement}}\label{algo:iterImprovement}
\end{algorithm2e}

\begin{lemma}\label{lem:algProvesIterImp}
Given a weighted graph $G$ and a feasible interface $\Phi=(I,T,\mathcal{C})$ of $G$, 
Algorithm~\ref{algo:iterImprovement} returns a $\Phi$-tour $F$ in $G$ with the guarantees stated in Theorem~\ref{thm:iterImprovement}.
\end{lemma}
\begin{proof}
The running time guarantee stated in Theorem~\ref{thm:iterImprovement} immediately follows from Theorem~\ref{thm:algShortTJoin} and Theorem~\ref{thm:algLongTJoin}, using $|I| + \frac{|T|}{\sfrac{\epsilon}{8}} \le \frac{9 |I|}{\epsilon}$.

Let $F\in \{F_1,F_2\}$ be the $\Phi$-tour returned by Algorithm~\ref{algo:iterImprovement}. 
To show that $F$ fulfills the approximation guarantee stated in~\eqref{eq:boostingGuarantee}, we distinguish two cases.

If $\ell(J) \leq \frac{\epsilon}{4}\cdot \ell(\OPT)$, then the solution $F_1$ will be short enough:
\begin{align*}
\ell(F) &\leq \ell(F_1)
\leq \left(1+\frac{\epsilon}{2}\right)\alpha\cdot\ell(\OPT) + (\alpha+1)\frac{\epsilon}{4}\cdot\ell(\OPT)
\leq (1+\epsilon)\cdot \alpha\cdot \ell(\OPT)\enspace,
\end{align*}
where we used $\sfrac{(\alpha+1)}{2}\leq \alpha$ for the last inequality, which holds because $\alpha\geq 1$.

If $\ell(J) \geq \frac{\epsilon}{4}\cdot \ell(\OPT)$, then the $\Phi$-tour $F_2$ will be short enough:
\begin{align*}
\ell(F) &\leq \ell(F_2)
\leq \left(\beta + \frac{\epsilon}{8}(\beta-1) \right)\cdot \ell(\OPT)
- (\beta - 1)\cdot \frac{\epsilon}{4}\cdot \ell(\OPT)
= \left( \beta - \frac{\epsilon}{8} (\beta-1) \right) \cdot \ell(\OPT)\enspace,
\end{align*}
thus completing the proof of Lemma~\ref{lem:algProvesIterImp}.
\end{proof}

For the proof of Theorem~\ref{thm:iterImprovement}, it remains to show
Theorem~\ref{thm:algShortTJoin} and Theorem~\ref{thm:algLongTJoin}.

\section{\boldmath Finding a short $\Phi$-tour if there is a short $T$-join}\label{sec:shortTJoin}

In this section we prove Theorem~\ref{thm:algShortTJoin}, i.e., how to get a short $\Phi$-tour if the shortest $T$-join has small length compared to $\ell(\OPT)$.

We start by analyzing a simple algorithm for computing a $\Phi$-tour. 
However, this simple algorithm will not be sufficient to prove Theorem~\ref{thm:algShortTJoin}.
Thus in a second step, we will refine the algorithm to obtain the desired bound.

\smallskip
\begin{algorithm2e}[H]
\KwIn{a weighted graph $G$, an interface $\Phi=(I,T,\Cscr)$ of $G$, and a $T$-join $J$ in $G$.}
\KwOut{an edge set $F$.}
\begin{enumerate}[label=\arabic*.]
\item In each connected component of $G$ apply $\mathcal{A}$ to get an $\alpha$-approximate TSP-tour.
Let $Q$ be the union of these tours.
\item Return $F = Q\cupp J$.
\end{enumerate}
\caption{A simple $\Phi$-TSP algorithm}\label{algo:shortTJoinNoGuessing}
\end{algorithm2e}
\smallskip

Notice that Algorithm~\ref{algo:shortTJoinNoGuessing} always returns an edge set $F$, even if the input is infeasible. 
We therefore show first that Algorithm~\ref{algo:shortTJoinNoGuessing} does return a $\Phi$-tour whenever it is run with a feasible input.
\begin{lemma}\label{lem:algoShortTJoinNoGuessingFeasible}
The set $F$ returned by Algorithm~\ref{algo:shortTJoinNoGuessing} is a $\Phi$-tour if and only if the input is feasible, i.e., $G$ admits a $\Phi$-tour.
\end{lemma}
\begin{proof}
Assume that $G$ admits a $\Phi$-tour, which implies by Lemma~\ref{lem:existencePhiTour} that the three properties 
\ref{item:PhiTourExistsJoin}, \ref{item:PhiTourExistsX}, and \ref{item:PhiTourExistsI} listed in Lemma~\ref{lem:existencePhiTour} are fulfilled. 
Because the set $Q$ computed in Algorithm~\ref{algo:shortTJoinNoGuessing} consists of TSP tours in each connected component of $G$, 
the vertex sets of the connected components of $(V,Q)$ and $G$ are the same. 
Because $G$ fulfills~\ref{item:PhiTourExistsX} and~\ref{item:PhiTourExistsI}, this implies that also $(V,Q)$ and $(V,Q\cupp J)$ fulfill these two conditions. 
Finally, $\odd(Q\cupp J) = \odd(J) = T$, because $\odd(Q)=\emptyset$ and $J$ is a $T$-join, which shows that $Q\cupp J$ is indeed a $\Phi$-tour.
\end{proof}

We now analyze the length of the $\Phi$-tour returned by Algorithm~\ref{algo:shortTJoinNoGuessing}.
\begin{lemma}\label{lem:shortTJoinNoGuessing}
Assume we are given an $\alpha$-approximation algorithm $\mathcal{A}$ for TSP.
Let $G=(V,E,\ell)$ be a weighted graph, $\Phi=(I,T,\Cscr)$ a feasible interface of $G$, and $J$ a $T$-join in $G$.
Then, Algorithm~\ref{algo:shortTJoinNoGuessing} computes a $\Phi$-tour $F$ in $G$ with 
\begin{equation*}
\ell(F)\leq \alpha \cdot \ell(\OPT) + (\alpha+1)\cdot \ell(J) + 2\alpha\cdot |I|\cdot \max\left\{\big. \ell(e) : e\in E\setminus (\OPT\cup J) \right\} \enspace,
\end{equation*}
in time $O(|V| \cdot f_{\mathcal{A}}(G))$,
where $\OPT$ is a shortest $\Phi$-tour. 
Here $\max\emptyset:=0$.
\end{lemma}

\begin{proof}
First, observe that the running time is indeed as claimed, because the bottleneck of the algorithm is calling $\mathcal{A}$ for each connected component of $G$; moreover, the connected components can be found in linear time and there are at most $|V|$ many of them.

To bound $\ell(Q)$, we transform a shortest $\Phi$-tour $\OPT$ into a union of TSP solutions, one for each connected component of $G$.
Let $L\subseteq E$ be a minimal edge set such that the vertex sets of
the connected components of $(V,\OPT\cup J\cup L)$ and $G$ are the same. 
Observe that the multi-set $\OPT\cupp J\cupp L \cupp L$ is a union of TSP solutions, one for each connected component of $G$. 
Because the set $Q$ determined in Algorithm~\ref{algo:shortTJoinNoGuessing} was obtained through $\mathcal{A}$, 
which is an $\alpha$-approximation algorithm, we have
\begin{equation*}
\ell(Q) \leq \alpha \cdot \left(\big.\ell(\OPT) + \ell(J) + 2 \ell(L) \right)\enspace,
\end{equation*}
and, hence, the solution $F=Q\cupp J$ returned by the algorithm satisfies
\begin{equation}\label{eq:lengthBoundOnFNoGuessing}
\ell(F) = \ell(Q) + \ell(J) \leq \alpha\cdot\ell(\OPT) + (\alpha+1)\cdot\ell(J) + 2 \alpha \cdot\ell(L)\enspace.
\end{equation}
Moreover, because $\OPT$ is a $\Phi$-tour, we have that $(V,\OPT)/I$ must be connected, which implies that 
$(V,\OPT)$ has at most $|I|$ connected components, and thus
\begin{equation*}
|L| \leq |I|-1\enspace.
\end{equation*}
Together with~\eqref{eq:lengthBoundOnFNoGuessing}, this leads to the desired guarantee:
\begin{align*}
\ell(F) &\leq \alpha \cdot\ell(\OPT) + (\alpha+1)\cdot\ell(J) + 2\alpha \cdot |I| \cdot \max\left\{\big. \ell(e) : e\in E\setminus (\OPT \cup J) \right\}\enspace,
\end{align*}
where the inequality follows from $L\subseteq E\setminus (\OPT \cup J)$, which holds because the edges in $L$ connect different connected components of $(V, \OPT\cup J)$.
\end{proof}

We now explain how to refine Algorithm~\ref{algo:shortTJoinNoGuessing} by a guessing step
to obtain the guarantees claimed in Theorem~\ref{thm:algShortTJoin}.
If all edges that are not contained in $\OPT \cup J$ have length at most $\frac{\delta \cdot \ell(\OPT)}{2 \cdot |I|}$,
Lemma~\ref{lem:shortTJoinNoGuessing} already implies the desired bound.
To obtain this property, we delete all edges from $G$ that are \emph{heavy}, 
i.e. have length at least $\frac{\delta \cdot \ell(\OPT)}{2 \cdot |I|}$, and are not contained in $\OPT \cup J$.
We guess this set of edges to delete as follows.
First we guess the set $H$ of heavy edges, which can be done in polynomial time by guessing a minimum length edge in $H$.
Then we guess the set $H^* = \OPT\cap H$ of heavy edges contained in $\OPT$.
Algorithm~\ref{algo:shortTJoin} formalizes this procedure and, as we show next, indeed implies Theorem~\ref{thm:algShortTJoin}.

\smallskip
\begin{algorithm2e}[H]
Compute a shortest $T$-join $J$ in $G$. \\
For $f\in E$ define $H_f \coloneqq \{ e\in E: \ell(e) \ge \ell(f) \}$.\\
\For{ every edge set $H \in \{ H_f : f\in E \} \cup \{ \emptyset \}$}{
   \For{every set $H^*\subseteq H$ with $|H^*|\leq \sfrac{2|I|}{\delta}$}{
      Set $D\coloneqq H \setminus (H^* \cup J)$. \\
      Apply Algorithm~\ref{algo:shortTJoinNoGuessing} to the graph $(V, E\setminus D)$ to obtain a multi-set $F_D$ of edges, which is a 
      $\Phi$-tour in $G$ if the input is feasible.
   }
}
Among all computed $\Phi$-tours $F_D$, return a shortest one.
\caption{$\Phi$-TSP algorithm to prove Theorem~\ref{thm:algShortTJoin}}\label{algo:shortTJoin}
\end{algorithm2e}
\smallskip

\begin{proof}[Proof of Theorem~\ref{thm:algShortTJoin}]
 We start by observing that the running time of Algorithm~\ref{algo:shortTJoin} is indeed bounded by $|V|^{O(\sfrac{|I|}{\delta})} \cdot f_{\mathcal{A}}(G)$.
 There are at most $|V|^2$  possible edges $f$ that are being considered in the outer for-loop.
 For each of them, there are $|V|^{O(\sfrac{|I|}{\delta})}$ possible sets $H^*$ considered in the inner for-loop.
 Thus, there are at most $|V|^{O(\sfrac{|I|}{\delta})}$ calls to Algorithm~\ref{algo:shortTJoinNoGuessing}.
 Finally, all other operations can be done in time $|V|^{O(1)}$.
  
 We now show that Algorithm~\ref{algo:shortTJoin} returns a $\Phi$-tour with the guarantees claimed by Theorem~\ref{thm:algShortTJoin}.
 Let $\OPT$ be a shortest $\Phi$-tour and let $H\coloneqq \{ e\in E : \ell(e) \ge \sfrac{\delta \cdot \ell(\OPT)}{2 |I|}\}$ be the 
 set of heavy edges. Then in some iteration of the outer for-loop we consider the set $H$.
 Because 
 \begin{equation*}
  \ell(\OPT) \ge \ell (H\cap \OPT) \ge |H\cap \OPT| \cdot  \frac{\delta \cdot \ell(\OPT)}{2 \cdot |I|}\enspace,
 \end{equation*}
 we have $|H\cap \OPT| \le \sfrac{2|I|}{\delta}$, and thus, we consider the set $H^* \coloneqq H\cap \OPT$ in some iteration of the inner for-loop.
 As $D = H \setminus (H^* \cup J)$ does not contain any edge of $\OPT$, the $\Phi$-tour $\OPT$ is a feasible solution 
 of the instance to which we apply Algorithm~\ref{algo:shortTJoinNoGuessing}.
 Moreover, the set $D$ contains all heavy edges not contained in $\OPT \cup J$ and hence by Lemma~\ref{lem:shortTJoinNoGuessing}, we obtain
 \begin{align*}
  \ell(F_D) &\le  \alpha \cdot \ell(\OPT) + (\alpha+1)\cdot \ell(J) + 2\alpha\cdot |I|\cdot \max \left\{\big. \ell (e) : e\in (E\setminus D) \setminus (\OPT \cup J) \right\} \\ 
  &\le   \alpha \cdot \ell(\OPT) + (\alpha+1)\cdot \ell(J) + 2\alpha\cdot |I|\cdot   \frac{\delta \cdot \ell(\OPT)}{2 \cdot |I|} \\
  &= (1+\delta)\cdot\alpha\cdot\ell(\OPT) + (\alpha+1)\cdot \ell(J)\enspace.
 \end{align*}
\end{proof}

\section{Iterative improvement via dynamic programming}\label{sec:dynProg}

In this section, we show how to prove Theorem~\ref{thm:algLongTJoin}, i.e., how to obtain a short $\Phi$-tour
if the length of a shortest $T$-join is large. 
Here, our goal is to use dynamic programming to ``guess'' a significant portion, in terms of total length, of edges used in $\OPT$.
Very recently, dynamic programming has become a strong tool in the context of Path TSP, Chain-Constrained Spanning Trees, 
and related problems~\cite{traub_2018_approaching,traub_2018_beating,zenklusen_2019_approximation,nagele_2019_new}, 
leading to the currently best known approximation factors for these settings. 
The dynamic programming idea we employ combines and extends elements used in these prior dynamic programming techniques. 

What we aim to achieve with dynamic programming in the context of $\Phi$-TSP, for some interface $\Phi=(I,T,\mathcal{C})$ of $G$, is the following. 
We can fix an arbitrary laminar family $\mathcal{L}$ of subsets of $V$. 
Our goal is to guess what edges of $\OPT$ are crossing the cuts in $\mathcal{L}$. 
Clearly, if $\OPT\cap \delta(L)$ contains many edges for some $L\in \mathcal{L}$, it seems computationally elusive to guess them. 
This is the reason why we fix some constant $k$ and only guess $\OPT$-edges in cuts $\delta(L)$ for $L\in \mathcal{L}$ if $|\OPT\cap \delta(L)|\leq k$. 
We denote the sets inducing these cuts by $\mathcal{L}(\OPT,k) \subseteq \mathcal{L}$ and the $\OPT$-edges in these cuts by $\OPT(\mathcal{L},k)\subseteq \OPT$. 
Formally, for any edge set $R\subseteq E$, we define
\begin{align*}
\mathcal{L}(R,k) &\coloneqq \left\{ L\in \mathcal{L} : |R\cap \delta(L)|\leq k \right\}\enspace, \text{ and}\\
R(\mathcal{L},k) &\coloneqq \bigcup_{L\in \mathcal{L}(R,k)} \left( \delta(L)\cap R
\right)\enspace.
\end{align*}
As we discuss in more detail later, guessing the edges $\OPT(\mathcal{L},k)$ can be achieved through a dynamic program that guesses the $\OPT$-edges 
in the different cuts defined by $\mathcal{L}$ step by step, from smaller to larger sets in $\mathcal{L}$.
However, the running time of the propagation step of the dynamic program depends on the number of disjoint sets in $\mathcal{L}$ 
that can be contained in some larger set $L\in\mathcal{L}$. 
We capture this dependency through the \emph{width} $\width(\mathcal{L})$ of the laminar family $\mathcal{L}$ 
(see~\cite{nagele_2019_new} for a similar use of this notion).

\begin{definition}[width of a laminar family]
The \emph{width} $\width(\mathcal{L})$ of a laminar family $\mathcal{L}$ is the number of minimal sets contained in the family.
\end{definition}
Observe that the number of minimal sets of a laminar family bounds the size of any subfamily of disjoint sets.

The following theorem formalizes what we can achieve through our dynamic program, which we present later in detail. 
Notice that for the algorithm to be efficient, we need $\mathcal{L}$ to have width bounded by a constant.
\begin{theorem}\label{thm:guessInLaminarFamily}
Let $\beta > 1$.
Assume there is a $\beta$-approximation algorithm $\mathcal{B}$ for $\Phi$-TSP. 
Then there is an algorithm that computes for any feasible interface $\Phi=(I,T,\mathcal{C})$ of a weighted graph $G=(V,E,\ell)$, 
any $k\in \mathbb{Z}_{\geq 0}$, and any laminar family $\mathcal{L}$ over $V$, a $\Phi$-tour $F$ with
\begin{equation}\label{eq:guessInLaminarFamily}
\ell(F) \leq \min\left\{\big.
\beta\cdot \ell(R) - (\beta -1 ) \cdot \ell(R(\mathcal{L},k)) : \text{$R$ is a $\Phi$-tour}
\right\}
\end{equation}
in time $|V|^{O\left(|I| + k\cdot \width(\mathcal{L})\right)}\cdot f_{\mathcal{B}} \left(\big. G, |I|+k\cdot (\width(\mathcal{L})+1) \right)$. 
In particular, the algorithm calls $\mathcal{B}$ only on instances with interfaces of size bounded by $|I|+k\cdot(\width(\mathcal{L})+1)$.
\end{theorem}
Note that the guarantee stated in~\eqref{eq:guessInLaminarFamily} for $R=\OPT$ indeed reflects the guessing of the edges in $\OPT(\mathcal{L},k)$. 
More precisely, by replacing $R$ by $\OPT$ in~\eqref{eq:guessInLaminarFamily}, we obtain a $\Phi$-tour $F$ with an upper bound on its length $\ell(F)$ that decomposes into two terms:
\begin{enumerate}[label=(\roman*),topsep=0.2em]
\item a term $\ell(\OPT(\mathcal{L},k))$, i.e., each edge $e\in \OPT(\mathcal{L},k)$ contributes its length $\ell(e)$, and

\item a term $\beta\cdot\ell(\OPT\setminus \OPT(\mathcal{L},k))$, where the length of each other edge in $\OPT$ 
gets inflated by the approximation factor $\beta$ of the algorithm $\mathcal{B}$.
\end{enumerate}

\subsection{Finding a suitable laminar family}

To make significant progress through Theorem~\ref{thm:guessInLaminarFamily}, we need to find a laminar family $\mathcal{L}$ over $V$ such that $\ell(\OPT(\mathcal{L},k))$ is large. 
Let $J$ be a shortest $T$-join.
If $\ell(J)$ is large, then we will construct a family $\mathcal{L}$ with the property that even for any $T$-join $R$, 
the length $\ell(R(\Lscr,k))$ is large. 
Notice that this implies what we want because $\OPT$ is a $T$-join.

This statement is formalized in Lemma~\ref{lem:goodLaminarFamily}, which is derived from the dual of the natural linear program to find a shortest $T$-join. 
We exploit that there is an optimal dual solution whose support corresponds to a laminar family of subsets of $V$, which follows from combinatorial uncrossing arguments. 

\begin{lemma}\label{lemma:dualTjoin}
Let $G=(V,E,\ell)$ be a weighted graph.
Moreover, let $T\subseteq V$ such that $G$ contains a $T$-join, and let $t \in T$.
Then there is strongly polynomial algorithm that computes a laminar family $\Lscr$ over $V \setminus \{t\}$ 
and values $y \in \mathbb{R}^{\mathcal{L}}_{> 0}$ such that
\begin{alignat}{3}
 \sum_{\substack{L\in\Lscr:\\ e\in\delta(L)}} y_L &\le \ell(e)\qquad &\forall e\in E\enspace,\label{eq:dualFeasibilityTJoin}
\\
\sum_{L\in\Lscr} y_L &= \ell(J), \text{ and}\quad &\label{eq:dualFeasibilityObjective}\\
|L \cap T| &\text{ is odd} &\forall L\in\Lscr\enspace.\label{eq:dualFeasibilityTCut}
\end{alignat}
\end{lemma}
\begin{proof}
We start with a classical linear description to find a minimum length $T$-join, based on the dominant of the $T$-join polytope. To this end, let $\mathcal{F} = \{Q\subseteq V\setminus \{t\} : |Q\cap T| \text{ is odd}\}$; these vertex sets induce all $T$-cuts. Then, the following linear program computes the value $\ell(J)$ of a shortest $T$-join (see, e.g.,~\cite{schrijver_2003_combinatorial}).
\begin{equation}\label{eq:domTJoinLP}
\begin{array}{rr@{\;}c@{\;}ll}
\min & \multicolumn{1}{c}{\displaystyle\sum_{e\in E}\ell(e) x_e}\\
&x(\delta(Q)) &\geq  &1 &\forall\; Q\in \mathcal{F}\\
&x &\in &\mathbb{R}^E_{\geq 0} &
\end{array}
\end{equation}
Its dual problem, which is a fractional $T$-cut packing problem, is given below.
\begin{equation}\label{eq:TCutPackingLP}
\begin{array}{rr@{\;}c@{\;}ll}
\max & \multicolumn{1}{c}{\displaystyle\sum_{Q\in \mathcal{F}} y_Q}\\
&\displaystyle\sum_{\substack{Q\in \mathcal{F}:\\ e\in \delta(Q)}}y_Q &\leq  &\ell(e) &\forall\; e\in E\\
&y &\in &\mathbb{R}^{\mathcal{F}}_{\geq 0} &
\end{array}
\end{equation}

If $y\in\mathbb{R}^{\mathcal{F}}_{\geq 0}$ is an optimum dual solution with laminar support $\Lscr$,
then $y$ and $\Lscr$ have the desired properties.
Here~\eqref{eq:dualFeasibilityObjective} follows from strong duality and~\eqref{eq:dualFeasibilityTCut} follows from $L\in \mathcal{F}$.

A strongly polynomial algorithm to compute such an optimal dual solution with laminar support can be obtained by standard techniques:
Using the framework of Frank and Tardos~\cite{frank_1987_application}, one can first find in strongly polynomial time a vector $\hat{\ell} \in \mathbb{R}^E_{\geq 0}$ with encoding length polynomial in $|E|$, and such that the set of optimal solutions of~\eqref{eq:domTJoinLP} remains the same when replacing $\ell$ by $\hat{\ell}$. Moreover, also the set of optimal dual bases remains the same. This allows for solving~\eqref{eq:domTJoinLP} in strongly polynomial time through the ellipsoid method. To find an optimal dual basis, one can delete all variables from~\eqref{eq:TCutPackingLP} that do not correspond to constraints encountered by the ellipsoid algorithm when solving~\eqref{eq:TCutPackingLP}. Now solving the reduced dual problem~\eqref{eq:TCutPackingLP} with $\hat{\ell}$ instead of $\ell$ allows for finding an optimal dual basis, which, by the result of Frank and Tardos, remains an optimal dual basis for~\eqref{eq:TCutPackingLP} without replacing $\ell$ by $\hat{\ell}$. Knowing an optimal dual basis, one can obtain an optimal solution to~\eqref{eq:TCutPackingLP} in strongly polynomial time by solving a linear equation system. 

Finally, this solution can be transformed into a laminar one by uncrossing:
if $y_A > 0$ and $y_B >0$ for $A,B\in \mathcal{F}$ with $A \setminus B \ne \emptyset$ and $B\setminus A \ne \emptyset$ and $A\cap B \ne \emptyset$,
then either $A\cap B$ and $A\cup B$ belong to $\mathcal{F}$ or $A \setminus B$ and $B\setminus A$ belong to $\mathcal{F}$;
we can increase the dual variables on these two sets by $\min\{ y_A, y_B\}$ and decrease the dual variables $y_A$ and $y_B$ by the same amount, maintaining
a feasible dual solution.
Karzanov~\cite{karzanov_1996_how} showed how to obtain a laminar family by a sequence of such uncrossing steps in strongly polynomial time.
\end{proof}

We now show that the family $\Lscr$ from Lemma~\ref{lemma:dualTjoin} has the desired properties.
\begin{lemma}\label{lem:goodLaminarFamily}
Let $G=(V,E,\ell)$ be a weighted graph.
Moreover, let $T\subseteq V$ such that $G$ admits a $T$-join. 
Then, there is a strongly polynomial algorithm that computes a laminar family $\mathcal{L}$ over $V$ with $\width(\mathcal{L})\leq \max \{ 0,\, |T| -1\}$ such that 
for any $T$-join $R\subseteq E$, and any $k\in \mathbb{Z}_{\geq 0}$, we have
\begin{align*}
\ell(R(\mathcal{L},k)) & \geq \ell(J) - \frac{1}{k+1}\cdot \ell(R)\enspace,
\end{align*}
where $J$ is a shortest $T$-join in $G$.
\end{lemma}
\begin{proof}
If $T=\emptyset$, we can simply set $\Lscr =\emptyset$ because $\ell(J)=0$.
Otherwise, we compute $\Lscr$ and $y$ as in Lemma~\ref{lemma:dualTjoin}
and show that $\Lscr$ has the desired properties.
Since every set in $\Lscr$ must contain an element of $T \setminus \{t\}$, we have $\width(\mathcal{L})\leq |T| -1$.

Let now $R\subseteq E$ be a $T$-join, and let $k\in \mathbb{Z}_{\ge 0}$.
Since $R$ is a $T$-join, it has a non-empty intersection with every cut $\delta(L)$ with $L\in\Lscr$ because of~\eqref{eq:dualFeasibilityTCut}.
Hence, by~\eqref{eq:dualFeasibilityTJoin},
\begin{align}\label{eq:lowerBoundEdgesToGuess}
  \ell(R(\mathcal{L},k)) &= \sum_{e\in R(\Lscr,k)} \ell(e) \geq \sum_{e\in R(\mathcal{L},k)} \sum_{\substack{L \in \Lscr:\\ e\in\delta(L)}} y_L 
  = \sum_{L \in \Lscr} |R(\mathcal{L},k) \cap \delta(L)| \cdot y_L  \ge \sum_{L\in \mathcal{L}(R,k)} y_L\enspace.
\end{align}
Again using~\eqref{eq:dualFeasibilityTJoin}, we moreover obtain
\begin{align}\label{eq:upperBoundEdgesNotToGuess}
 \ell(R)&\ge \sum_{e\in R} \sum_{\substack{L \in \Lscr:\\ e\in\delta(L)}} y_L 
 = \sum_{L\in\Lscr} |R \cap \delta(L)| \cdot y_L %
 \ge \sum_{L\in \Lscr \setminus \mathcal{L}(R,k)} (k+1) \cdot y_L\enspace.
\end{align}
Combining~\eqref{eq:dualFeasibilityObjective},~\eqref{eq:lowerBoundEdgesToGuess}, and~\eqref{eq:upperBoundEdgesNotToGuess}, we obtain
\begin{align*}
\ell(J) & = \sum_{L\in\Lscr} y_L = \sum_{L\in \mathcal{L}(R,k)} y_L + \sum_{L\in \Lscr \setminus \mathcal{L}(R,k)} y_L
\le  \ell(R(\mathcal{L},k)) + \frac{1}{k+1} \cdot \ell(R)\enspace,
\end{align*}
as desired.
\end{proof}

Finally, Theorem~\ref{thm:algLongTJoin} is a direct consequence of Theorem~\ref{thm:guessInLaminarFamily} and Lemma~\ref{lem:goodLaminarFamily}.

\begin{proof}[Proof of Theorem~\ref{thm:algLongTJoin}]
If $T =\emptyset$, we simply call the given $\beta$-approximation algorithm $\Bscr$.
Otherwise, let $k=\lfloor\sfrac{1}{\delta}\rfloor$.
We apply Lemma~\ref{lem:goodLaminarFamily} to obtain in strongly polynomial time a laminar family $\mathcal{L}$ over $V$ such that
\begin{enumerate}
\item $\displaystyle\ell(R(\mathcal{L},k)) \geq \ell(J) - \frac{1}{k+1}\cdot \ell(R) \geq \ell(J) - \delta \cdot \ell(R) \quad \forall\;\text{$T$-join $R\subseteq E$}$, and
\item $\width(\mathcal{L}) \leq \max\{ 0,\, |T|-1\}=|T|-1$, where the equality follows from the assumption $T\neq\emptyset$.
\end{enumerate}
\smallskip
Because a shortest $\Phi$-tour $\OPT$ is a $T$-join, we have
\begin{equation*}
\ell(\OPT(\mathcal{L},k)) \geq \ell(J) -\delta\cdot \ell(\OPT)\enspace,
\end{equation*}
which, together with Theorem~\ref{thm:guessInLaminarFamily} implies the desired results, i.e., that one can find a $\Phi$-tour $F$ in $G$ with
\begin{align*}
\ell(F) &\leq
 \beta \cdot \ell(\OPT) - (\beta -1) \cdot \ell(\OPT(\mathcal{L},k))  \\
&\leq \beta\cdot \ell(\OPT) - (\beta-1)\left(
\ell(J) - \delta\cdot \ell(\OPT)
\right)
\\
&= (\beta + \delta\cdot (\beta-1))\cdot \ell(\OPT) - (\beta-1)\cdot \ell(J)\enspace,
\end{align*}
in time
\begin{equation*}
|V|^{O\left(|I| + \frac{\width(\mathcal{L})}{\delta}\right)}
\cdot f_{\mathcal{B}}\left(G, |I|+ \frac{\width(\mathcal{L})+1}{\delta}\right)
\leq |V|^{O\left(|I| +\frac{|T|}{\delta}\right)} \cdot f_{\mathcal{B}}\left(G, |I| +\frac{|T|}{\delta}\right)\enspace.
\end{equation*}

\end{proof}

It remains to derive Theorem~\ref{thm:guessInLaminarFamily}, which, as mentioned, we show through a dynamic programming approach.

\subsection{Combining partial solutions}\label{sect:combining_solutions}

In the analysis of our dynamic programming algorithm we use the following notion of an \emph{induced interface}, 
which allows us to analyze the algorithm with respect to interfaces coming from a shortest $\Phi$-tour $\OPT$.

\begin{definition}[induced interface]\label{def:induced_interface} 
Let $G=(V,E,\ell)$ be a weighted graph.
Let $\Phi=(I,T,\mathcal{C})$ be an interface of $G$, and let $F$ be a $\Phi$-tour in $G$.
For $W \subseteq V$, the interface $\Phi_W=(I_W,T_W,\mathcal{C}_W)$ induced by $(F,\Phi)$ on $W$ is defined by 
\begin{enumerate}[label=(\roman*)]
\item\label{item:induced_interface_i_w} $I_W = (I\cap W) \cup U$, where $U$ is the set of vertices in $W$ that are connected by an edge of $F$ to a vertex in $V\setminus W$,
\item $T_W=\odd(F[W])$, and
\item $\mathcal{C}_W\subseteq 2^{I_W}$ contains, for each connected component of $(W,F[W])$, a set including all vertices of $I_W$ contained in that connected component.
\end{enumerate}

\end{definition}

See Figure~\ref{fig:indInterface} for an example of an induced interface. Moreover, also Figure~\ref{fig:guess3}, which we used as an illustrative example in the introduction to showcase the guessing of multiple edges per cut, highlights an induced interface with $W=L_3\setminus L_1$, which is induced by an $s$-$t$ tour. 
We remark that the interface induced by $(F, \Phi)$ depends only on $F$ and $I$, not on $T$ or $\mathcal{C}$.

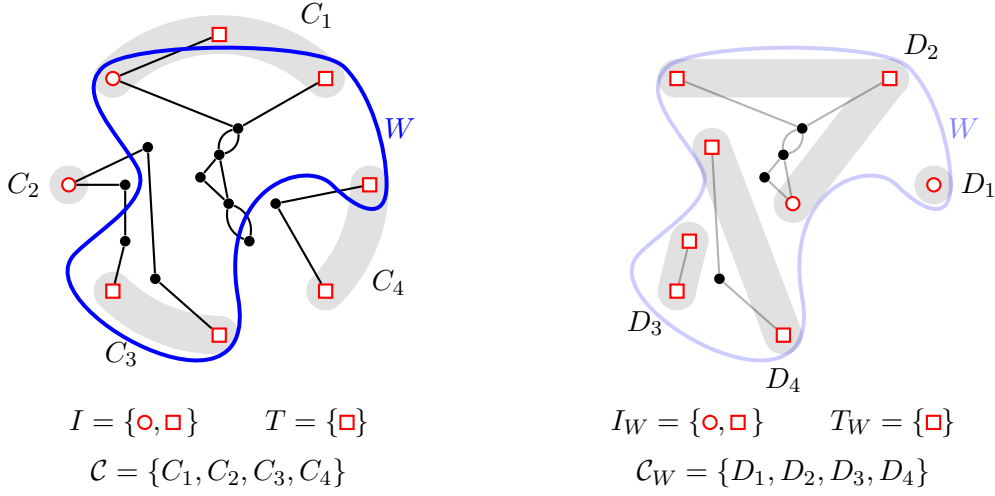
\begin{figure}
\begin{center}
\begin{tikzpicture}[scale =0.5]

\pgfdeclarelayer{fg}
\pgfdeclarelayer{bg}
\pgfsetlayers{bg,main,fg}

\def\rad{4}
\def\cd{0.5}
\def \bs{40pt}

\tikzset{nsT/.style={
draw=red,thick,rectangle,inner sep=0em,minimum size=5pt,fill=white
}}

\tikzset{nsI/.style={
fill=white, draw=red,thick,circle,inner sep=0em,minimum size=5pt
}}

\tikzset{ns/.style={
fill=black,circle,inner sep=0em,minimum size=4pt}
}

\newcommand\Cset[2]{
\draw (#1:\rad+\cd) arc (#1:#2:\rad+\cd)
  arc (#2:#2+180:\cd)
  arc (#2:#1:\rad-\cd)
  arc (#1+180:#1+360:\cd)
;
}

\newcommand\CsetSeg[2]{
\pgfmathanglebetweenpoints{\pgfpointanchor{#1}{center}}{\pgfpointanchor{#2}{center}}
\global\let\myangle\pgfmathresult

\draw ($(#1)+(\myangle-270:\cd)$)
  arc (\myangle-270:\myangle-90:\cd)
  -- ($(#2)+(\myangle-90:\cd)$)
  arc (\myangle-90:\myangle+90:\cd);

}

\begin{scope}

\begin{scope}[every node/.style={nsI}]
\node[nsT] (v1) at (0:\rad) {};
\node[nsT] (v2) at (-45:\rad) {};

\node[nsT] (v3) at (-90:\rad) {};
\node[nsT] (v4) at (-135:\rad) {};

\node (v5) at (180:\rad) {};

\node (v6) at (135:\rad) {};
\node[nsT] (v7) at (90:\rad) {};
\node[nsT] (v8) at (45:\rad) {};
\end{scope}

\begin{scope}[every path/.style={fill=black!12,draw=black!12}]
\begin{pgfonlayer}{bg}

\Cset{-45}{0}
\Cset{-135}{-90}
\Cset{180}{180}
\Cset{45}{135}
\end{pgfonlayer}
\end{scope}

\begin{scope}
\def\td{0.7}
\node at (60:\rad+\cd+\td) {$C_1$};
\node at (180:\rad+\cd+\td) {$C_2$};
\node at (-120:\rad+\cd+\td) {$C_3$};
\node at (-30:\rad+\cd+\td) {$C_4$};
\end{scope}

\begin{scope}[every node/.style={ns}, shift={(-9.5,-5.5)}]
 \node (n1) at (7,4) {};
 \node (n2) at (7,5.5) {};
 \node (n3) at (7.6,6.5) {};
 \node (n5) at (7.8,3) {};

 \node (n8) at (10,7) {};
 \node (n4) at (9.5,6.3) {};
 \node (n9) at (9,5.7) {};
\node (n10) at (9.75,5) {};
\node (n11) at (10.3,4) {};

 \node (n6) at (11,5) {};

\end{scope}

\begin{scope}[thick, shift={(-9.5,-5.5)}]
\draw (v3) -- (n5) -- (n3);
\draw (v4) -- (n1);
\draw (n1) -- (n2);
\draw (n8) to[bend right=\bs] (n4);
\draw (n8) to[bend left=\bs] (n4);

\draw (n4) -- (n9);
\draw (n4) -- (n10);
\draw (n9) -- (n10);
\draw (n10) to[bend left=\bs] (n11);
\draw (n10) to[bend right=\bs] (n11);

\draw (n2) -- (v5) -- (n3) ;

\draw (v2) -- (n6) -- (v1);

\draw (v7) -- (v6) -- (n8) -- (v8);
\end{scope}

\begin{scope}[line width=1.5pt,blue]
\coordinate (ws) at ($(v1)+(0.45,0)$);
\draw (ws) .. controls +(-90:1) and +(-45:1) .. ($(v1)-(1,0)$)
.. controls +(135:1) and +(100:3) .. ($(v3)+(0.5,1)$)
.. controls +(-80:3.8) and +(-135:2) .. ($(v4)+(-1,1)$)
.. controls +(45:1) and +(-60:1) .. ($(n2)!0.5!(n3)$)
.. controls +(120:1) and +(225:1) .. (135:\rad+0.4)
.. controls +(45:1) and +(135:1) .. ($(v8)+(0.5,0.3)$)
.. controls +(-45:1) and +(90:1) .. (ws);
\end{scope}

\end{scope} 

\node[blue] at ($(v1)+(0.8,1.5)$) {$W$};

\begin{scope}[shift={(0,-7)}]

\node[align=center] {$I = \{\tikz{\node[nsI] {};},\tikz{\node[nsT] {};}\, \}$
\qquad $T=\{\tikz{\node[nsT] {};}\}$\\[0.5em]
$\mathcal{C}=\{C_1,C_2,C_3,C_4\}$
};

\end{scope}

\begin{scope}[xshift=15cm]

\begin{scope}[every node/.style={nsI}]
\node (v1) at (0:\rad) {};

\node[nsT] (v3) at (-90:\rad) {};
\node[nsT] (v4) at (-135:\rad) {};

\node[nsT] (v6) at (135:\rad) {};
\node[nsT] (v8) at (45:\rad) {};
\end{scope}

\begin{scope}[every node/.style={ns}, shift={(-9.5,-5.5)}]
 \node[nsT] (n1) at (7,4) {};
 \coordinate (n2) at (7,5.5);
 \node[nsT] (n3) at (7.6,6.5) {};
 \node (n5) at (7.8,3) {};

 \node (n8) at (10,7) {};
 \node (n4) at (9.5,6.3) {};
 \node (n9) at (9,5.7) {};
\node[nsI] (n10) at (9.75,5) {};

\end{scope}

\begin{scope}[line width=1.5pt,blue,opacity=0.2]
\coordinate (ws) at ($(v1)+(0.45,0)$);
\draw (ws) .. controls +(-90:1) and +(-45:1) .. ($(v1)-(1,0)$)
.. controls +(135:1) and +(100:3) .. ($(v3)+(0.5,1)$)
.. controls +(-80:3.8) and +(-135:2) .. ($(v4)+(-1,1)$)
.. controls +(45:1) and +(-60:1) .. ($(n2)!0.5!(n3)$)
.. controls +(120:1) and +(225:1) .. (135:\rad+0.4)
.. controls +(45:1) and +(135:1) .. ($(v8)+(0.5,0.3)$)
.. controls +(-45:1) and +(90:1) .. (ws);
\end{scope}

\node[blue,opacity=0.5] at ($(v1)+(0.8,1.5)$) {$W$};

\begin{scope}[every path/.style={fill=black!12,draw=black!12}]
\begin{pgfonlayer}{bg}

\Cset{0}{0}
\CsetSeg{v8}{n10}
\CsetSeg{v8}{v6}

\CsetSeg{n1}{v4}
\CsetSeg{n3}{v3}

\end{pgfonlayer}
\end{scope}

\begin{scope}
\def\td{0.7}
\node at (0:\rad+\cd+\td) {$D_1$};
\node at (45:\rad+\cd+\td) {$D_2$};
\node at (225:\rad+\cd+\td) {$D_3$};
\node at (270:\rad+\cd+\td) {$D_4$};
\end{scope}

\begin{scope}[thick, shift={(-9.5,-5.5)},opacity=0.3]
\draw (v3) -- (n5) -- (n3);
\draw (v4) -- (n1);
\draw (n8) to[bend right=\bs] (n4);
\draw (n8) to[bend left=\bs] (n4);

\draw (n4) -- (n9);
\draw (n4) -- (n10);
\draw (n9) -- (n10);

\draw (v6) -- (n8) -- (v8);
\end{scope}

\begin{scope}[shift={(0,-7)}]
\node[align=center] {$I_W = \{\tikz{\node[nsI] {};},\tikz{\node[nsT] {};}\, \}$
\qquad $T_W=\{\tikz{\node[nsT] {};}\}$\\[0.5em]
$\mathcal{C}_W=\{D_1,D_2,D_3,D_4\}$
};
\end{scope}

\end{scope} 

\end{tikzpicture}
 \end{center}
\caption{On the left-hand side, an interface $\Phi=(I,T,\mathcal{C})$ on a graph $G=(V,E)$ is shown together with a $\Phi$-tour $F\subseteq E$ (the black edges) and a set $W\subseteq V$. The right-hand side figure depicts the interface $\Phi_W=(I_W,T_W,\mathcal{C}_W)$ induced by $(F,\Phi)$ on $W$.}\label{fig:indInterface}
\end{figure}

The following lemma shows some basic properties of induced interfaces. 

\begin{lemma}\label{lem:induced_interfaces}
 Let $G=(V,E,\ell)$ be a weighted graph and $\Phi=(I,T,\mathcal{C})$ an interface of $G$.
 Let $F$ be a $\Phi$-tour in $G$ and $W\subseteq V$.
 Let $\Phi_W$ be the interface induced by $(F,\Phi)$ on $W$.
 Then 
 \begin{enumerate}
  \item \label{item:induced_interface_is_interface} $\Phi_{W}$ is an interface of $G[W]$,
  \item \label{item:induced_tours_feasible_for_induced_interfaces} $F[W]$ is a $\Phi_W$-tour in $G[W]$, and
  \item \label{item:subtour-induces-same-interface} for every $W' \subseteq W$,
        the interface induced by $(F[W],\Phi_W)$ on $W'$ equals the interface induced by $(F,\Phi)$ on $W'$.
 \end{enumerate}
\end{lemma}
\begin{proof}
Let $\Phi_W=(I_W,T_W,\mathcal{C}_W)$.
As in Definition~\ref{def:induced_interface}~\ref{item:induced_interface_i_w}, let $U$ be  the set of vertices in $W$ that are connected by an edge of $F$ to a vertex in $V\setminus W$.

To prove~\ref{item:induced_interface_is_interface}, we have to observe that $T_W\subseteq I_W$. (Notice that we clearly have that $|T_W|$ is even because $T_W = \odd(F[W]).)$
Let $u\in T_W$. 
If $F$ contains an edge connecting $u$ with $V\setminus W$, then $u\in U$
and hence $u\in I_W$.
Otherwise we have $(\delta(u)\cap F)\subseteq E[W]$ and hence $u\in T_W = \odd(F[W])$ implies $u\in \odd(F)$.
Since $F$ is a $\Phi$-tour, we conclude $u\in T \subseteq I$.
Moreover, $u\in T_W=\odd(F[W])\subseteq W$, so $u \in I \cap W \subseteq I_W$.

To prove  \ref{item:induced_tours_feasible_for_induced_interfaces}, we have to show that $(W,F[W])/I_W$ is connected
(the other two conditions of Definition~\ref{def:phi_tour} trivially hold).
Suppose not. Then there is a set $W' \subseteq W \setminus I_W$ with $W' \ne W$ and $F[W] \cap \delta(W') = \emptyset$.
This implies, together with $I_W=(I\cap W)\cup U$---which holds by definition of $I_W$---that $W' \subseteq V\setminus I$ with $W' \ne V$ and $F\cap \delta(W') = F[W] \cap \delta(W') = \emptyset$.
This contradicts the fact that $(V,F)/I$ is connected, which has to hold because $F$ is a $\Phi$-tour.

To show~\ref{item:subtour-induces-same-interface}, let $(I_1, T_1, \Cscr_1)$ be the interface induced by $(F, \Phi)$ on $W'$
and let $(I_2, T_2, \Cscr_2)$ be the interface induced by $(F[W], \Phi_W)$ on $W'$.
Let $U_1$ be the set of vertices in $W'$ that are connected by an edge of $F$ to a vertex in $V\setminus W'$.
Let $U_2$ be the set of vertices in $W'$ that are connected by an edge of $F[W]$ to a vertex in $W\setminus W'$.
Then $U_1 = (U\cap W') \cup U_2$. Therefore,
\begin{align*}
 I_1 &= (I \cap W') \cup U_1 \\
 &= (I\cap W') \cup (U\cap W') \cup U_2 \\
 &=  \left(\big. ((I\cap W) \cup U) \cap W' \right) \cup U_2 \\
 &= (I_{W} \cap W') \cup U_2 \\
 &= I_2\enspace.
\end{align*}
Finally, because $(F[W])[W']=F[W']$, which follows from $W'\subseteq W$, we have
\begin{equation*}
T_2 = \odd((F[W])[W']) = \odd(F[W']) = T_1\enspace,
\end{equation*}
and also $\mathcal{C}_1 = \mathcal{C}_2$, because these partitions of $I_1=I_2$ are both defined with respect to the connected components of $(W',F[W'])$, because $(W', (F[W])[W'])=(W',F[W'])$.
\end{proof}

Notice that given an interface $\Phi=(I,T,\mathcal{C})$ on a graph $G=(V,E)$ and a $\Phi$-tour $F\subseteq E$, then the interface $\Phi_V$ induced by $(F,\Phi)$ on $V$ is not necessarily identical to $\Phi$. More precisely, $\Phi_V=(I_V,T_V,\mathcal{C}_V)$ always fulfills $I_V=I$ and $T_V=T$. However, $F$ may connect different parts of the partition $\mathcal{C}$, which, in the interface $\Phi_V$, will then only appear as one set in $\mathcal{C}_V$. See the left-hand side illustration in Figure~\ref{fig:indInterface} for such an example where the highlighted $\Phi$-tour would induce an interface $\Phi_V\neq \Phi$ on $V$ because $C_2\cup C_3$ is a single set in $\mathcal{C}_V$.

In our dynamic program we will combine solutions for different subgraphs with induced interfaces. 
The following lemma shows sufficient conditions under which this works out.

\begin{lemma}\label{lem:composedSubInterfaces}
Let $G=(V,E,\ell)$ be a weighted graph.
Let $\Phi=(I,T,\mathcal{C})$ be an interface of $G$ and let $F$ be a $\Phi$-tour in $G$. 
Let $W_0,\ldots, W_p$ be a partition of $V$. 
For $i\in \{0,\ldots,p\}$, let  $\Phi_i=(I_i,T_i,\mathcal{C}_i)$ be the interface induced by $(F,\Phi)$ on $W_i$, 
and let $F_i$ be a $\Phi_i$-tour in $G[W_i]$. Then
\begin{equation*}
F' \coloneqq X \cupp \left(\bigcup_{i=0}^p F_i \right)
\end{equation*}
is a $\Phi$-tour in $G$, where $X \coloneqq F\cap \bigcup_{i=0}^p \delta(W_i)$.
\end{lemma}
\begin{proof}
We first show point~\ref{item:PhiTourIsTJoin} of Definition~\ref{def:phi_tour}, i.e.\  $\odd(F') =T$.
For $i\in \{0,\ldots,p\}$, we have $\odd(F_i) = T_i = \odd(F[W_i])$ since $F_i$ is a $\Phi_i$-tour.
Thus
\begin{align*}
 \odd(F') &= \odd(X) \symdiff \odd(F_0) \symdiff \cdots \symdiff \odd(F_p) \\
 &= \odd(X) \symdiff \odd(F[W_0]) \symdiff \cdots \symdiff \odd(F[W_p]) \\
 &= \odd(F) \\
 &= T\enspace,
\end{align*}
where $\symdiff$ denotes the symmetric difference; 
we used $F = X \cupp F[W_0] \cupp \dots \cupp F[W_p]$.
Before proving that $F'$ also fulfills the remaining two properties of a $\Phi$-tour, we show the following claim.
See Figure~\ref{fig:combining_phi_tours} for an illustration.

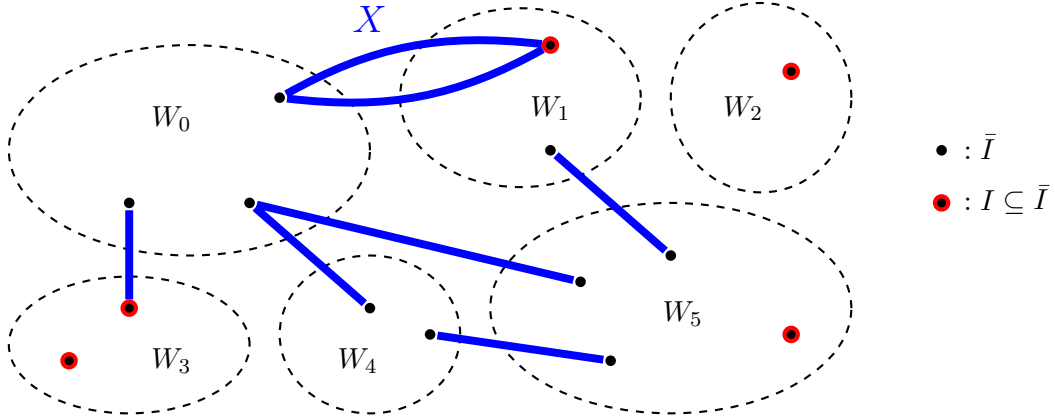
\begin{figure}
\begin{center}
 \begin{tikzpicture}[yscale=0.7, xscale=0.8]

\tikzset{oneEdge/.style={line width=3pt,blue}}
\def \bs{20pt}

\tikzset{nsU/.style={
fill=black,draw=none, ultra thick,circle,inner sep=0em,minimum size=4pt
}}

\tikzset{nsI/.style={
fill=black,draw=red, ultra thick,circle,inner sep=0em,minimum size=5pt
}}

\node[blue] (X) at (6,7.5) {\Large $X$};

\begin{scope}[xshift=15.5cm]
\node[nsI] (ILab) at (0,4) {};
\node[right] at ($(ILab)+(0.2,0.06)$) {$: I \subseteq \bar I$};

\node[nsU] (ULab) at (0,5) {};
\node[right] at ($(ULab)+(0.2,0.06)$) {$:\bar I$};
\end{scope}

\begin{scope}[dashed, thick]
 \node (e1) at (2.7,5.6) {$W_0$};
 \draw (3,5) ellipse (3 and 2);
 \node (e2) at (9,5.8) {$W_1$};
 \draw (8.5,6) ellipse (2 and 1.7);
  \node (e2) at (12.2,5.8) {$W_2$};
 \draw (12.5,6) ellipse (1.5 and 1.8);
 \node (e3) at (2.7,1.0) {$W_3$};
 \draw (2,1.3) ellipse (2 and 1.3);
  \node (e4) at (5.8,1.0) {$W_4$};
 \draw (6,1.5) ellipse (1.5 and 1.5);
  \node (e5) at (11.2,1.9) {$W_5$};
 \draw (11,2) ellipse (3 and 2);
\end{scope}

\begin{scope}[every node/.style={nsI}]
\node (v1) at (2,2) {};
\node (v4) at (1,1) {};
\node (v5) at (13, 1.5) {};
\node (v6) at (9,7) {};
\node (v7) at (13,6.5) {};
\end{scope}

\begin{scope}[every node/.style={nsU}]
\node (w1) at (6,2) {};
\node (w2) at (2,4) {};
\node (w3) at (4,4) {};
\node (w4) at (4.5, 6) {};
\node (w5) at (9.5, 2.5) {};
\node (w6) at (11,3) {};
\node (w8) at (7, 1.5) {};
\node (w9) at (9,5) {};
\node (w10) at (10,1) {};
\end{scope}

\begin{scope}[oneEdge]
 \draw (v6) to[bend right=\bs] (w4);
 \draw (v6) to[bend left=\bs] (w4);
 \draw (w2) -- (v1);
 \draw (w5) -- (w3) -- (w1);
 \draw (w6) -- (w9);
 \draw (w8) -- (w10);
\end{scope}

 \end{tikzpicture}  \end{center}
\caption{ 
\label{fig:combining_phi_tours}
The dashed ellipsoids show the partition of $V$ into $W_0,\dots, W_p$.
The thick blue edges are the edges in $X$.
Only the vertices in $\bar I$ are shown here, where the vertices with a red boundary
are those contained in $I$.
}
\end{figure}

\begin{claim}\label{claim:connectivity_preserved}
 Let $\overline{I} \coloneqq I_0 \cupp \dots \cupp I_p$ and $a,b \in \overline{I}$.
 Suppose $(V, F)$ contains an $a$-$b$~path. Then $(V,F')$ contains an $a$-$b$~path.
\end{claim}

\begin{proof}[Proof of Claim~\ref{claim:connectivity_preserved}]
Suppose the claim is wrong. Then there exist vertices $a,b\in \overline{I}$ such that
$(V, F)$ contains an $a$-$b$~path $P$, but $(V,F')$ does not.
We choose $a$, $b$, and $P$ such that the number of edges of $P$ is minimum.
Consequently, $P$ contains no vertex of $\overline{I} \setminus \{ a, b\}$.
We now distinguish two cases. 

\noindent 
\textbf{Case 1:} $X\cap E(P) = \emptyset$. \\
Then $P$ is completely contained in a single set $W_i$ for some $i\in \{0,\ldots,p\}$, by definition of $X$.
Hence, $a,b \in W_i \cap \overline{I} = I_i$ and $a$ and $b$ are connected by the path $P$ in $(W_i, F[W_i])$.
Since $\Phi_i$ is the interface induced by $(F,\Phi)$ on $W_i$, the vertices $a$ and $b$ are contained in the same set of the partition $\Cscr_i$ of $I_i$.
This implies that every $\Phi_i$-tour, and in particular $F_i$, must contain an $a$-$b$~path, contradicting the assumption 
that $(V,F')$ contains no $a$-$b$~path.

\noindent 
\textbf{Case 2:} $X\cap E(P) \ne \emptyset$. \\
Recall $X = F\cap \bigcup_{i=0}^p \delta(W_i)$. 
For $i\in\{0,\ldots,p\}$, the set $I_i$ contains all vertices of $W_i$ that are an endpoint of an edge in $X$, by definition of the induced interface $\Phi_i$.
Thus all endpoints of edges in $X$ are contained in $\overline{I}$.
Since $P$ contains no vertex of $\overline{I} \setminus \{ a, b\}$, we have $X \cap E(P) = \{\{a,b\}\}$, i.e.
the path $P$ consists only of a single edge that is contained in $X$ and thus also in $F'$.
This contradicts our assumption that $(V,F')$ contains no $a$-$b$~path.
\qedhere\\[-0.3em]
\mbox{}\hfill {\footnotesize(proof of Claim~\ref{claim:connectivity_preserved})}
\end{proof}

To show point~\ref{item:PhiTourConnectivity} of Definition~\ref{def:phi_tour},
we need to show that any two vertices $a$ and $b$ that are contained in the same set of the partition $\Cscr$ of $I$
are also contained in the same connected component of $(V,F')$.
If $a$ and $b$ are contained in the same set of the partition $\Cscr$, they are contained in the same connected component of $(V,F)$
because $F$ is a $\Phi$-tour.
Hence by Claim~\ref{claim:connectivity_preserved} and $I\subseteq \bar I$, also $(V, F')$ contains an $a$-$b$~path.

It remains to show point~\ref{item:PhiTourConnectToInterface} of Definition~\ref{def:phi_tour}, i.e., we prove that $(V,F')/ I$ is connected. 
First observe that if $p=0$, then the result holds because then $F'=F_0$ is a $\Phi_0$-tour and $I_0=I$. 
Hence, assume from now on $p>0$. In this case, we first observe that
\begin{equation}\label{eq:IiNotEmpty}
I_i \neq \emptyset \quad\forall i\in \{0,\ldots,p\}\enspace.
\end{equation}
Indeed, because $(V,F)/I$ is connected, which follows from $F$ being a $\Phi$-tour, we have for each $i\in \{0,\ldots,p\}$ that either $I\cap W_i\neq \emptyset$ or $\delta(W_i)\cap F\neq\emptyset$, both of which imply $I_i\neq \emptyset$.

To conclude that $(V,F')/I$ is connected, we will observe the following two properties, which immediately imply the result:
\begin{enumerate}[label=(\alph*)]
\item\label{item:WiConnToIi} For each $i\in \{0,\ldots,p\}$, each vertex $v\in W_i$ is connected to a vertex in $I_i$ in the graph $(W_i,F_i)$.
\item\label{item:IiConn} All vertices in $\cup_{i=0}^p I_i$ are connected in $(V,F')/I$.
\end{enumerate}
Notice that~\ref{item:WiConnToIi} is a consequence of~\eqref{eq:IiNotEmpty} and the fact that $(W_i,F_i)/I_i$ is connected, which holds because
$F_i$ is a $\Phi_i$-tour in $G[W_i]$.
 Finally,~\ref{item:IiConn} follows from Claim~\ref{claim:connectivity_preserved} due to the following. Either $I=\emptyset$, in which case $(V,F)/I = (V,F)$ is connected---because $F$ is a $\Phi$-tour---which implies \ref{item:IiConn} by Claim~\ref{claim:connectivity_preserved}. Or $I\neq\emptyset$, in which case the connectivity of $(V,F)/I$ implies that in $(V,F)$ each vertex $v\in \cup_{i=0}^p I_i$ is connected to a vertex of $I$, again implying~\ref{item:IiConn}  by Claim~\ref{claim:connectivity_preserved}.
\end{proof}

\subsection{The dynamic program}

We now expand on the dynamic program used to show Theorem~\ref{thm:guessInLaminarFamily}.
The dynamic program is formally described by Algorithm~\ref{algo:guessInLaminarFamily} below. 
See also Figure~\ref{fig:dynamic_program} for an illustration.
Before formally proving that Algorithm~\ref{algo:guessInLaminarFamily} indeed returns a $\Phi$-tour implying Theorem~\ref{thm:guessInLaminarFamily}, we provide a brief explanatory discussion outlining the core ideas of the algorithm and the line of reasoning we employ to show its correctness.
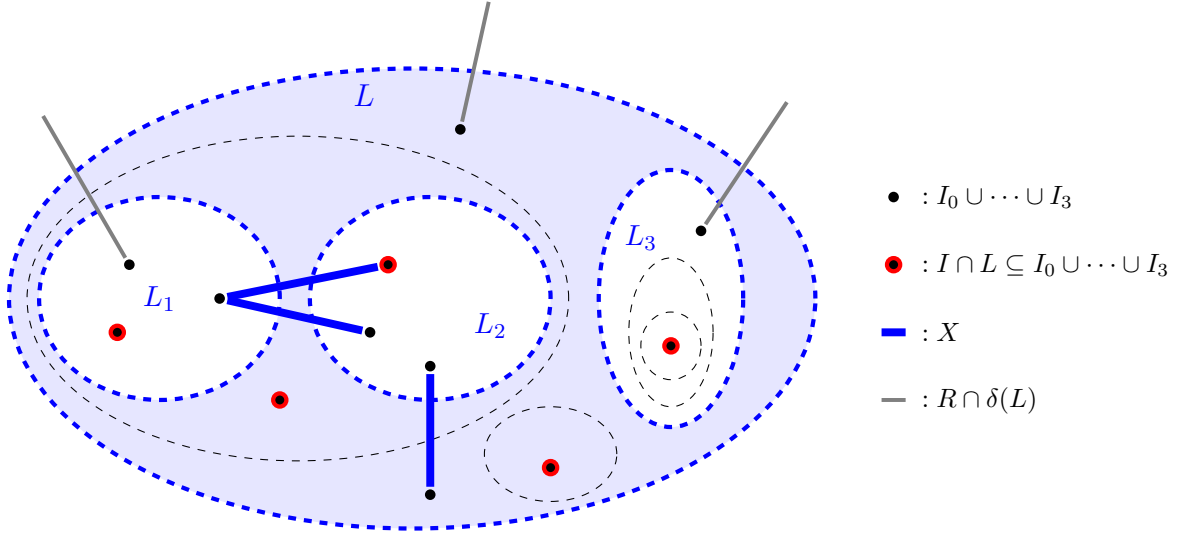
\begin{figure}[b!]
\begin{center}
 \begin{tikzpicture}[yscale=0.9, xscale=0.8]

\tikzset{oneEdge/.style={line width=3pt,blue}}
\tikzset{grayOneEdge/.style={line width=1.5pt,gray}}

\def \bs{20pt}

\tikzset{nsN/.style={
fill=none,draw=none, ultra thick,circle,inner sep=0em,minimum size=4pt, outer sep=1pt
}}

\tikzset{nsU/.style={
fill=black,draw=none, ultra thick,circle,inner sep=0em,minimum size=4pt, outer sep=1pt
}}

\tikzset{nsI/.style={
fill=black,draw=red, ultra thick,circle,inner sep=0em,minimum size=5pt, outer sep=1.5pt
}}

\begin{scope}
\node[nsI] (ILab) at (14.7,4) {};
\node[right] at ($(ILab)+(0.3,0.0)$) {\small $: I \cap L \subseteq I_0 \cup \dots \cup I_3 $};

\node[nsU] (ULab) at (14.7,5) {};
\node[right] at ($(ULab)+(0.3,0.0)$) {\small $: I_0 \cup \dots \cup I_3 $};

\draw[oneEdge] (14.5,3) -- (14.9,3);
\node[right] at (15,3) {\small $: X$};

\draw[grayOneEdge] (14.5,2) -- (14.9,2);
\node[right] at (15,2) {\small $: R\cap \delta(L)$};
\end{scope}

\begin{scope}[dashed, ultra thick, blue]
 \node (L) at (5.9, 6.5) {\large$L$};
 \draw[fill=blue, fill opacity=0.1] (6.7,3.5) ellipse (6.7 and 3.4);
 
 \draw[fill=white] (2.5,3.5) ellipse (2 and 1.5);
  \node(L1) at (2.5, 3.5) {$L_1$};
 
 \draw[fill=white] (7,3.5) ellipse (2 and 1.5);
   \node(L2) at (8, 3.1) {$L_2$};
 
 \draw[fill=white] (11,3.5) ellipse (1.2 and 1.9);
  \node(L3) at (10.5, 4.4) {$L_3$};
\end{scope}

\begin{scope}[dashed]
 \draw (11, 3) ellipse (0.7 and 1.1);
 \draw (11, 2.8) ellipse (0.5 and 0.5);
 \draw (4.8,3.5) ellipse (4.5 and 2.4);
 \draw (9, 1.2) ellipse (1.1 and 0.7);
\end{scope}

\begin{scope}[every node/.style={nsI}]
\node (t1) at (11,2.8) {};
\node (t2) at (1.8,3) {};
\node (t3) at (4.5,2) {};
\node (t4) at (9,1) {};
\node (t5) at (6.3,4) {};
\end{scope}

\begin{scope}[every node/.style={nsU}]
\node (v1) at (2,4) {};
\node (v2) at (3.5,3.5) {};
\node (v3) at (7.5, 6) {};
\node (v4) at (7,2.5) {};
\node (v5) at (11.5, 4.5) {};
\node (v6) at (6,3) {};
\node (v7) at (7,0.6) {};
\end{scope}

\begin{scope}[every node/.style={nsN}]
\node (o1) at (0.5,6.3) {};
\node (o2) at (8,8) {};
\node (o3) at (13, 6.5) {};
\end{scope}

\begin{scope}[oneEdge]
 \draw (v2) -- (v6);
 \draw (v4) -- (v7);
 \draw (t5) -- (v2);
\end{scope}

\begin{scope}[grayOneEdge]
\draw (o1) -- (v1);
\draw (o2) -- (v3);
\draw (o3) -- (v5);
\end{scope}

 \end{tikzpicture}  \end{center}
\caption{
Illustration of Algorithm~\ref{algo:guessInLaminarFamily}.
The dashed ellipses show the laminar family $\Lscr$; 
these sets are considered by the algorithm in an order of non-decreasing cardinality.
Suppose we are considering $L\in \Lscr(R,k)$; only subsets of $L$ are shown in the figure.
In the dynamic program we guess the children $L_1, \dots, L_p$ of $L$ in the laminar family  $\Lscr(R,k)$.
The sets $L_1, \dots, L_p$ are shown as blue ellipses with white interior.
The light blue area shows the set $L_0 = L \setminus (L_1 \cup \dots \cup L_p)$.
\\
We also guess the set $X$ of edges in $R[L] \cap (\delta(L_1) \cup \dots \cup \delta (L_p))$; these are the thick blue edges.
The thin gray edges are the edges in $R\cap \delta(L)$; these will be guessed only in a later step.
However, we do guess the interface $\Phi_L =(I_L, T_L, \Cscr_L)$ that $(R,\Phi)$ induces on $L$, where $I_L$ consists of 
all vertices in $I\cap L$ (shown with a thick red boundary) and all vertices in $L$ that are endpoints of gray edges.
Moreover, we guess the interfaces $\Phi_i =(I_i, T_i, \Cscr_i)$ that $(R,\Phi)$ induces on the sets $L_i$ for $i\in \{0,\dots,p\}$.
The picture shows only the vertices in $I_0 \cup \dots \cup I_p$;
these are the vertices in $L$ that are contained in the set $I$ or are an endpoint of a thick blue or thin gray edge.
\\
We compute a $\Phi_0$-tour in $G[L_0]$ and combine it with $X$ and the $\Phi_i$-tours in $G[L_i]$ for $i\in\{1,\dots,p\}$
that we have computed in previous steps of the dynamic program.
This yields a $\Phi_L$-tour in $G[L]$.
\label{fig:dynamic_program}
}
\end{figure}

To this end, let $R$ be a $\Phi$-tour (unknown to the algorithm), and we will show that the dynamic program returns a $\Phi$-tour $F\subseteq E$ such that $\ell(F) \leq \beta \cdot \ell(R) - (\beta-1)\cdot \ell(R(\mathcal{L},k))$.
Conceptually, we want to consider the elements of the laminar family $\Lscr(R,k) \subseteq \Lscr$ from smaller to larger ones.
Since we do not know the laminar family $\Lscr(R,k)$, we consider all sets in $\Lscr$ in an arbitrary fixed order of non-decreasing cardinality.
We then guess, for every vertex set $L\in\Lscr(R,k)$, the interface $\Phi_L$ induced by $(R,\Phi)$ on $L$.
Now we compute a $\Phi_L$-tour $F_{L,\Phi_L}$ in $G[L]$ as follows.

First, we guess the children $L_1, \dots, L_p$ of $L$ in the laminar family $\Lscr(R,k)$.
Then we guess the set $X \subseteq R[L]$ of edges that cross the cuts $\delta(L_1), \dots , \delta(L_p)$.
In other words, we guess all edges in  $R(\Lscr,k)$ that are contained in $L$, but not in any child of $L$.
Moreover, for each child $L_i$ with $i\in \{1,\ldots,p\}$, we guess the interface $\Phi_i$ induced by $(R,\Phi)$ on $L_i$.
Because we consider the elements of the laminar family $\Lscr$ in an order of non-decreasing cardinality, we have already considered 
$L_i$ before considering the current set $L$.
Hence we have already computed some $\Phi_i$-tour $F_{L_i,\Phi_i}$ for all $i\in \{1,\ldots,p\}$.

We now want to extend the union of these $\Phi_i$-tours for all $i\in\{1,\ldots,p\}$ and the set $X$ of edges crossing the 
boundaries of the children $L_1, \dots, L_p$ to a $\Phi_L$-tour in $G[L]$.
To this end we define $L_0  \coloneqq L\setminus \cup_{i=1}^p L_i$. 
Then $L_0, \dots, L_p$ is a partition of $L$.
We also guess the interface $\Phi_0$ that $(R,\Phi)$ induces on  $L_0$.
Then, by Lemma \ref{lem:composedSubInterfaces} applied to the graph $G[L]$, 
the union of $X$ and arbitrary $\Phi_i$-tours in $G[L_i]$ for $i = \{0, \dots, p\}$ is a $\Phi_L$-tour in $G[L]$.
Here we use that $\Phi_i$ is the interface induced by $(R[L], \Phi_L)$ on $L_i$ for $i =\{0, \dots, p\}$ (cf.~Lemma~\ref{lem:induced_interfaces}
\ref{item:subtour-induces-same-interface}).
Finally, we use the given algorithm $\Bscr$ to compute a $\beta$-approximation  $F_0$ 
of a minimum length $\Phi_0$-tour in the subgraph $G[L_0]$ and 
combine $X$, $F_0$, and the $\Phi_i$-tours $F_{L_i,\Phi_i}$ for $i\in \{1,\ldots,p\}$ to a $\Phi_L$-tour $F_{L,\Phi_L}$.

\smallskip
\begin{algorithm2e}[H]
\setstretch{1.05}

\SetKwData{Nil}{Nil}

Let $\overline{\mathcal{L}}= \mathcal{L}\cup \{V\}$\;

\For{$L\in \overline{\mathcal{L}}$, in non-decreasing order of cardinality}{

\For{each interface $\Phi_L=(I_L,T_L,\mathcal{C}_L)$ of $G[L]$ with $|I_L|\leq |I| + k$}{

Let $F_{L,\Phi_L} \coloneqq \Nil$\tcp*{No $\Phi_L$-tour found yet.}
\tcp*[f]{We set by convention: $\ell(\Nil)=\infty$.}

\For{each subfamily $\{L_1,\ldots, L_p\}\subseteq \overline{\mathcal{L}}$ of disjoint proper subsets of $L$}{

Let $L_0 \coloneqq L\setminus \cup_{i=1}^p L_i$\;

\For{all $X\subseteq \left(\bigcup_{i=1}^p \delta(L_i)\right) \cap E[L]$ with $|X\cap \delta(L_i)|\leq k \;\;\forall i\in \{1,\ldots,p\}$ }{

For $i\in \{0, \dots, p\}$, let $U_i \subseteq L_i$ be the set of all vertices in $L_i$ that are an endpoint of some edge in $X$\;
Let $I_i \coloneqq (I_L\cap L_i) \cup U_i \;\;\forall i \in \{0,\ldots, p\} $\;

\For{all $T_i, \mathcal{C}_i$ such that  $\Phi_i=(I_i, T_i, \mathcal{C}_i)$ is an interface of $G[L_i] \;\;\ (i \in \{0,\ldots, p\})$ }{

Use Algorithm $\mathcal{B}$ to find a $\Phi_0$-tour $F_0$\;

Let $F\coloneqq X\cupp F_0 \cupp \bigcup_{i=1}^p F_{L_i,\Phi_i}$\;

\If{$F$ is a $\Phi_L$-tour and $\ell(F)\leq \ell(F_{L,\Phi_L})$}{

Set $F_{L,\Phi_L} = F$\;

}

}

}

}

}

}

\Return $F_{V,\Phi}$\;

\caption{Algorithm computing a $\Phi$-tour with the properties mentioned in Theorem~\ref{thm:guessInLaminarFamily}}\label{algo:guessInLaminarFamily}
\end{algorithm2e}

\bigskip

In what follows, we now provide a rigorous proof that Algorithm~\ref{algo:guessInLaminarFamily} implies Theorem~\ref{thm:guessInLaminarFamily} by leveraging the tools from Section~\ref{sect:combining_solutions}.

\subsection{Proof of Theorem~\ref{thm:guessInLaminarFamily}}

We start by showing that Algorithm~\ref{algo:guessInLaminarFamily} has indeed the claimed running time, before proving its correctness.

\subsubsection*{Running time}

The running time of Algorithm~\ref{algo:guessInLaminarFamily} is dominated by the $5$-fold nested for-loops. 
We first determine upper bounds on the number of iterations of each for-loop separately, whenever the algorithm reaches it.

\begin{enumerate}[label=\raisenth* for-loop:, leftmargin=2.6cm]

\item It goes over all sets in $\overline{\mathcal{L}}$. 
Because $\overline{\mathcal{L}}$ is a laminar family over $V$, it contains $O(|V|)$ sets.

\item It goes over all interfaces $\Phi_L=(I_L, T_L, \mathcal{C}_L)$ of $G[L]$ with $|I_L| \leq |I| + k$. 
There are no more than $(|L|+1)^{|I| +k} \leq (|V|+1)^{|I| +k}$ choices for choosing $I_L$. 
Moreover, there are at most $2^{|I_L|} \leq 2^{|I|+k}$ choices for $T_L\subseteq I_L$. 
Finally, the number of partitions $\mathcal{C}_L$ of $I_L$ can be upper bounded by $|I_L|^{|I_L|} \leq |V|^{|I|+k}$. 
Overall, the number of iterations of any run of the second for-loop is bounded by $|V|^{O(|I|+k)}$.

\item It iterates over subfamilies of $\overline{\mathcal{L}}$ of disjoint proper subsets of $L$. 
Because the sets are disjoint, such a family can have at most $\width(\overline{\mathcal{L}}) \le \width(\mathcal{L}) + 1$ sets, 
and we can therefore bound the number of these subfamilies by $|\overline{\mathcal{L}}|^{\width(\mathcal{\overline{L}})} = |V|^{O(\width(\mathcal{L}))}$.

\item It iterates over edge sets $X\subseteq (\cup_{i=1}^p \delta(L_i)) \cap E[L]$ with $|X\cap \delta(L_i)| \leq k$ for all $i\in \{1,\ldots,p\}$, and can be bounded as follows. 
Notice that $|X| \leq \sum_{i=1}^p |X\cap \delta(L_i)| \leq p\cdot k \leq \width(\mathcal{L}) \cdot k$. 
Hence, there are at most $(|E|+1)^{k\cdot \width(\mathcal{L})} = |V|^{O(k\cdot \width(\mathcal{L}))}$ options for $X$.

\item 
This loop runs for all $i\in \{0,1,\dots,p\}$ over all interfaces $\Phi_i=(I_i,T_i,\mathcal{C}_i)$ of $G[L_i]$ , 
where $L_i$ and $I_i$ are fixed.
The number of interfaces $\Phi_i$ for a fixed $i\in \{0, \dots, p\}$ is thus bounded by $(2|I_i|)^{|I_i|} \leq (2|V|)^{|I_i|}$ and, 
hence, the total number of combinations of such interfaces, and thus also on the number of iterations each time this for-loop is run, is bounded by
\begin{align}\label{eq:boundOnIterations5}
\prod_{i=0}^p (2|V|)^{|I_i|}
= (2|V|)^{\sum_{i=0}^p |I_i|}\enspace.
\end{align}
Moreover, for $i\in \{1,\ldots,p\}$, we have $ |I_i| \leq  k + |I_L \cap L_i|$, which follows from the fact that each set $I_i$ contains the elements of $I_L \cap L_i$ 
together with at most $k$ endpoints of edges from $X$ because $|X\cap \delta(L_i)|\leq k$.
This implies
\begin{equation}\label{eq:boundSizeOfIi}
 \sum_{i=1}^p |I_i|  \leq p \cdot k + |I_L| \le \width(\mathcal{L})\cdot k + (|I| + k)  = O(|I|+k\cdot \width(\mathcal{L}))\enspace.
\end{equation}
Similarly,  
\begin{equation}\label{eq:boundSizeOfI0}
|I_0| \leq |I_L| + k \cdot \width(\mathcal{L}) \leq |I| + k + k\cdot \width(\mathcal{L}) = O(|I|+k\cdot \width(\mathcal{L}))\enspace.
\end{equation}
Combining \eqref{eq:boundSizeOfIi} and \eqref{eq:boundSizeOfI0} with \eqref{eq:boundOnIterations5},
we can bound the number of iterations of the fifth for-loop by $|V|^{O(|I|+k\cdot \width(\mathcal{L}))}$.
\end{enumerate}

\smallskip

The most expensive single operation performed by Algorithm~\ref{algo:guessInLaminarFamily} is the call to Algorithm~$\mathcal{B}$ to find a $\Phi_0$-tour, 
which, by assumption, takes no more than $f_{\mathcal{B}}(G, |I_{0}|)$ time. 
Due to the bound on $|I_0|\leq |I| + k\cdot (\width(\mathcal{L})+1)$ provided by~\eqref{eq:boundSizeOfI0}, we have that the total running time is thus indeed bounded by
$|V|^{O(|I|+k\cdot \width(\mathcal{L}))} \cdot f_{\mathcal{B}}(G, |I| + k \cdot (\width(\mathcal{L})+1))$.

\subsubsection*{Correctness}

We now show that, whenever $G$ admits a $\Phi$-tour, then Algorithm~\ref{algo:guessInLaminarFamily} 
will find a $\Phi$-tour $F_{V,\Phi}$ with the length guarantee claimed by Theorem~\ref{thm:guessInLaminarFamily}.
So let $R$ be a $\Phi$-tour. We have to show that $F_{V,\Phi}$ computed by the algorithm is a $\Phi$-tour (instead of \textsf{Nil}) and that it satisfies
\begin{equation}\label{eq:FPhiGood}
\ell(F_{V,\Phi}) \leq \beta \cdot \ell(R\setminus R(\mathcal{L},k)) + \ell(R(\mathcal{L},k))\enspace.
\end{equation}
We prove~\eqref{eq:FPhiGood} by showing the following claim from smaller to larger sets $L\in \mathcal{L}(R,k)\cup \{V\}$.

\begin{claim}\label{claim:subinterfacesAreGood}
Let $L\in \mathcal{L}(R,k) \cup \{V\}$.
If $L=V$, let $\Phi_L = \Phi$. 
Otherwise, let $\Phi_L=(I_L, T_L, \mathcal{C}_L)$ be the interface induced by $(R,\Phi)$ on $L$. 
Then Algorithm~\ref{algo:guessInLaminarFamily} computes a $\Phi_L$-tour $F_{L,\Phi_L}$ such that
\begin{equation*}
\ell(F_{L,\Phi_L}) \leq \beta\cdot \ell\left(R[L]\setminus R(\mathcal{L},k)\right) + \ell\left( R[L] \cap R(\mathcal{L},k) \right)\enspace.
\end{equation*}
\end{claim}

Observe that the claim immediately implies Theorem~\ref{thm:guessInLaminarFamily} by choosing $L=V$.
Hence, it remains to prove the claim.

\begin{proof}[Proof of Claim~\ref{claim:subinterfacesAreGood}]

We prove the claim by induction from smaller to larger sets in $\mathcal{L}(R,k)\cup \{V\}$. 
Hence, let $L\in \mathcal{L}(R,k)\cup \{V\}$ and assume that the claim holds for sets in $\mathcal{L}(R,k)\cup \{V\}$ of strictly smaller cardinality than $L$. 
In particular, it holds for the children $L_1, \ldots, L_p$ of $L$ in the laminar family $\mathcal{L}(R,k)\cup \{V\}$. (Note that $L$ may also not have any children.)
Let $L_0 \coloneqq L \setminus \cup_{i=1}^p L_i$, and for $i\in \{0,\ldots, p\}$, let $\Phi_i = (I_i, T_i, \mathcal{C}_i)$ be the interface induced by $(R,\Phi)$ on $L_i$. 
By using Lemma~\ref{lem:induced_interfaces} \ref{item:subtour-induces-same-interface} in the case $L\ne V$, 
we observe that $\Phi_i$ is also the interface induced by $(R[L], \Phi_L)$ on $L_i$.
Let $F_0$ be a $\Phi_0$-tour obtained through Algorithm~$\mathcal{B}$. Because $L_0, L_1,\ldots, L_p$ partitions $L$, we have by Lemma~\ref{lem:composedSubInterfaces} that
\begin{equation*}
F\coloneqq X
\cupp F_0
\cupp\; \bigcup_{i=1}^p F_{L_i,\Phi_i} 
\end{equation*}
is a $\Phi_L$-tour, where
\begin{equation}\label{eq:setXToBeConsidered}
X \coloneqq R[L]\cap \bigcup_{i=1}^p \delta(L_i)\enspace.
\end{equation}
 Before discussing that this $\Phi_L$-tour $F$ will indeed be considered by Algorithm~\ref{algo:guessInLaminarFamily}, we bound its length.
 First, $\ell(F_0) \le \beta \cdot \ell(R[L_0])$ because $\mathcal{B}$ is a $\beta$-approximation algorithm and
 $R[L_0]$ is a $\Phi_0$-tour by  Lemma \ref{lem:induced_interfaces}~\ref{item:induced_tours_feasible_for_induced_interfaces}.
 Moreover, for $i\in \{1,\dots, p\}$ we apply the induction hypothesis to $L_i$ and $\Phi_i$, which is possible because $L_i \in \Lscr(R,k)$ has strictly smaller cardinality than $L$.
 Hence, $F_{L_i,\Phi_i}$ is a  $\Phi_i$-tour and fulfills the length bound stated in the claim.
 We therefore get
\begin{align}
\ell(F) &= \ell\left( X \right) + \ell(F_0) + \sum_{i=1}^p \ell\left( F_{L_i, \Phi_i} \right)\notag\\
&\leq \ell\left( X \right) + \beta\cdot \ell(R[L_0]) + \sum_{i=1}^p \big(
\beta\cdot \ell\left( R[L_i] \setminus R(\mathcal{L},k)) + \ell(R[L_i] \cap R(\mathcal{L},k) \right)
\big)\notag\\
&= \beta\cdot \ell\left( R[L] \setminus R(\mathcal{L},k) \right) + \ell\left( R[L] \cap R(\mathcal{L},k) \right)\label{eq:lengthOfCompFGood}\enspace,
\end{align}
where the last equality follows by observing that
\begin{align*}
R[L_0],\; R[L_1]\setminus R(\mathcal{L},k), \ldots, R[L_p]\setminus R(\mathcal{L},k)\quad &\text{partitions}\quad R[L]\setminus R(\mathcal{L},k)\enspace, \text{ and}\\
X,\;
R[L_1]\cap R(\mathcal{L},k),\ldots,R[L_p]\cap R(\mathcal{L},k)\quad
&\text{partitions}\quad R[L]\cap R(\mathcal{L},k)\enspace.
\end{align*}
Due to~\eqref{eq:lengthOfCompFGood}, the $\Phi_L$-tour $F$ fulfills the length bound of the claim. 
It remains to show that the $\Phi_L$-tour $F$ will indeed be considered by Algorithm~\ref{algo:guessInLaminarFamily}.
For this, we show that the following quantities are considered in the five nested for-loops:
\begin{enumerate}[label=\raisenth* for-loop:, leftmargin=2.6cm]
\item considers $L$,
\item considers the interface $\Phi_L=(I_L,T_L,\mathcal{C}_L)$,
\item considers the children $L_1,\ldots, L_p$ of $L$ in the laminar family $\mathcal{L}(R,k)\cup \{V\}$,
\item considers the set $X$,
\item considers, for $i\in \{0,\dots, p\}$, the interfaces $\Phi_i$ induced by $(R[L],\Phi_L)$ on $L_i$.
\end{enumerate}
\smallskip
This run would indeed produce $F$. All that remains to be shown is that the above five quantities, to be considered within the five nested for-loops, fulfill the conditions set by the respective for-loops:

\begin{enumerate}[label=\raisenth* for-loop:, leftmargin=2.6cm]
\item Algorithm~\ref{algo:guessInLaminarFamily} considers all sets in $\mathcal{L}$ and hence, also $L$.

\item 
If $L = V$, the interface $\Phi$ is obviously considered.
Otherwise $\Phi_L=(I_L, T_L, \mathcal{C}_L)$ is the interface induced by $(R,\Phi)$ on $L$,
and we have $I_L = (I \cap L) \cup U$, where $U$ is the set of vertices in $L$ connected by an edge of $R$ to a vertex in $V\setminus L$. 
As $L \in \Lscr(R,k) \cup \{V\}$, we have $|\delta(L)\cap R|\leq k$, and hence $|U|\leq k$, 
which implies $|I_L| \leq |I| + k$ and shows that the interface $\Phi_L$ is considered in the second for-loop.

\item We have $\{L_1,\ldots, L_p\} \subseteq \overline{\mathcal{L}}$ . Hence, the subfamily $\{L_1, \ldots, L_p\}$ will be considered in the third nested for-loop.

\item The set $X$ we want to consider is given by~\eqref{eq:setXToBeConsidered}. 
This set clearly satisfies $X\subseteq \left(\cup_{i=1}^p\delta(L_i)\right) \cap E[L]$ because $R[L]\subseteq E[L]$. Moreover, for each $i\in \{1,\ldots,p\}$ we have
\begin{equation*}
|X\cap \delta(L_i)| = |R[L]\cap \delta(L_i)| \leq |R\cap \delta(L_i)| \leq k\enspace,
\end{equation*}
where the last inequality follows from $L_i\in \mathcal{L}(R,k)$. Hence, the set $X$ will be considered during the fourth nested for-loop of the algorithm.

\item For $i\in \{0,\ldots, p\}$ we have that $\Phi_i=(I_i,T_i,\mathcal{C}_i)$ is the interface of $G[L_i]$ induced by $(R[L],\Phi_L)$ on $L_i$. 
Hence, $I_i = (I_L \cap L_i) \cup U_i$, where $U_i$ are all vertices in $L_i$ connected by an edge of $R[L]$ to a vertex in $L\setminus L_i$.
We have $R[L]\cap \delta(L_i) = X \cap \delta(L_i)$ by our choice of $X$ as described in~\eqref{eq:setXToBeConsidered} and because $\{L_0,\dots, L_p\}$ is a partition of $L$.
Therefore, $I_i \coloneqq (I_L\cap L_i) \cup U_i$, as desired. 
Hence, the interfaces $\Phi_i$ for $i\in \{0,\dots, p\}$ indeed get considered in the fifth nested for-loop of the algorithm.
\end{enumerate}
\end{proof}

As said, Claim~\ref{claim:subinterfacesAreGood} implies \eqref{eq:FPhiGood}, completing the
proof of Theorem~\ref{thm:guessInLaminarFamily}.

We remark that Claim~\ref{claim:subinterfacesAreGood} can be slightly strengthened as follows. The statement also holds when replacing the induced interface $\Phi_L = (I_L,T_L,\mathcal{C}_L)$ by any interface $\Phi_L' = (I_L, T_L, \mathcal{C}_L')$ where $\mathcal{C}_L'$ is a refinement of $\mathcal{C}_L$. However, we do not need this for our purposes.

\section{Proof of the main theorem}\label{sec:MainThm}

We finally prove that the Boosting Theorem (Theorem~\ref{thm:iterImprovement}) implies Theorem~\ref{thm:mainReduction}. 
In fact, we prove a generalization, stated below as Theorem~\ref{thm:phiTourMain}, which, for $k=2$ and $\Phi=(I,T,\Cscr)$ with $I=T=\{s,t\}$ and $\Cscr=\{\{s,t\}\}$, yields Theorem~\ref{thm:mainReduction}.

\begin{theorem}\label{thm:phiTourMain}
    Let $\mathcal{A}$ be an $\alpha$-approximation algorithm for TSP. Then, for any $\epsilon >0$ and any integer $k$, there is an $(\alpha+\epsilon)$-approximation algorithm for $\Phi$-TSP restricted to instances with $|I_{\Phi }|\le k$ that, for any instance $(G,\Phi)$, calls $\mathcal{A}$ a strongly polynomial number of times on TSP instances defined on subgraphs of $G$, and performs further operations taking strongly polynomial time.
\end{theorem}

\begin{proof}
\leavevmode
We obtain the result by repeatedly applying the Boosting Theorem, i.e., Theorem~\ref{thm:iterImprovement}, to strengthen the $4$-approximation algorithm for $\Phi$-TSP guaranteed by Theorem~\ref{thm:fourApprox} through the $\alpha$-approximation algorithm for TSP which we assume to exist.
Without loss of generality $\epsilon \le 1$.
The Boosting Theorem will be repeated $i_{\max}$ many times with error parameter given by $\epsilon'=\sfrac{\epsilon}{\alpha}$, where
\begin{equation*}
i_{\max} := \left\lceil \frac{4- (\alpha+\epsilon)}{\alpha -1} \cdot \frac{8 \alpha}{\epsilon} \right\rceil\enspace.
\end{equation*}
Notice that $i_{\max}$ is constant, because both $\epsilon$ and $\alpha$ are fixed.

Let $\beta_0 := 4$ be the approximation factor for $\Phi$-TSP before applying the Boosting Theorem. We assume $\alpha \le 1.5 <\beta_0$ because Christofides' algorithm is a strongly polynomial $1.5$-approximation algorithm for TSP \cite{christofides_1976_worst-case, serdukov_1978_some}.
Let $i\in \{1, \dots, i_{\max}\}$.
After $i$ applications of the Boosting Theorem we obtain an algorithm $\mathcal{B}_{i}$ for $\Phi$-TSP with approximation ratio at most 
\begin{equation*}
\beta_i \coloneqq \max \left\{(1 + \epsilon')\alpha ,\ \beta_{i-1} - \frac{\epsilon'}{8}\cdot(\beta_{i-1} -1) \right\}
= \max \left\{\alpha + \epsilon ,\ \beta_{i-1} - \frac{\epsilon}{8 \alpha}\cdot(\beta_{i-1} -1) \right\}
\enspace,
\end{equation*} 
where we used $\epsilon' = \sfrac{\epsilon}{\alpha}$.
We therefore have 
\begin{align*}
 \beta_i &\le \max \left\{\alpha + \epsilon ,\ \beta_{i-1} - \frac{\epsilon}{8\alpha}(\alpha -1) \right\} 
 \le \max \left\{\alpha + \epsilon,\ \beta_{0} - i \cdot \frac{\epsilon}{8\alpha}(\alpha -1) \right\}\enspace,
\end{align*}
where the last inequality follows by induction on $i$.
Hence, 
\begin{align*}
 \beta_{i_{\max}} &\le \max \left\{\alpha + \epsilon ,\ 4 - i_{\max} \cdot \frac{\epsilon}{8\alpha}(\alpha -1) \right\} \\
 &= \max \left\{\alpha + \epsilon,\ 4 - \left\lceil \frac{4- (\alpha + \epsilon)}{\alpha -1} \cdot \frac{8\alpha}{\epsilon} \right\rceil \cdot \frac{\epsilon}{8\alpha}(\alpha -1) \right\} \\
 &= \alpha + \epsilon\enspace .
\end{align*}

Moreover, we define real numbers $k_i>0$ for $i\in \{0,\ldots, i_{\max}\}$ to upper bound the size of the interfaces we have to be able to handle after $i$ boosting steps. We want the $\beta_{i_{\max}}$-approximation algorithm $\mathcal{B}_{i_{\max}}$, obtained after $i_{\max}$ many applications of the Boosting Theorem, to handle interfaces of size $k_{i_{\max}}\coloneqq k$. Because $\mathcal{B}_{i_{\max}}$ was obtained by applying the Boosting Theorem to $\mathcal{B}_{i_{\max}-1}$, we obtain that $\mathcal{B}_{i_{\max}-1}$ needs to handle interfaces of size bounded by $k_{i_{\max}-1}\coloneqq \frac{9}{\epsilon'} \cdot k_{i_{\max}}$. Repeating this reasoning, we obtain upper bounds $k_i$ on the size of the interfaces that we have to handle with $\mathcal{B}_i$ that satisfy
\begin{equation*}
k_i \coloneqq \frac{9}{\epsilon'} \cdot k_{i+1} \quad \forall i\in \{i_{\max}-1,\ldots, 1,0\}\enspace,
\end{equation*}
which implies
\begin{equation*}
k_i= k\cdot \left(\frac{9}{\epsilon'}\right)^{i_{\max}-i} \qquad \forall i\in \{0,\ldots, i_{\max}\}\enspace.
\end{equation*}
Notice that because $i_{\max}$, $k$, and $\epsilon'$ are constant, also $k_0$ is constant.

For $i=i_{\max}$, the following claim implies Theorem~\ref{thm:phiTourMain} because $\beta_{i_{\max}}= \alpha + \epsilon$ and $i_{\max}$, $k_0$, and $\epsilon'$ are constant and 
$\Bscr_0$ is a strongly polynomial algorithm.
\begin{claim}\label{claim:inductionIterativeImprovement}
Let $c>0$ be the hidden constant in the big-$O$ notation in the runtime bound in Theorem~\ref{thm:iterImprovement}.
Let $i\in \{0,\dots, i_{\max}\}$ and let $\Ascr$ be the given $\alpha$-approximation algorithm for TSP.
Then there is a $\beta_i$-approximation algorithm $\Bscr_i$ for $\Phi$-TSP that, for every weighted graph $G$, runs in time at most
\begin{equation}\label{eq:claimIndIterImp}
 f_i(G) \coloneqq |V|^{i \cdot c \cdot\frac{k_0}{\epsilon'}} \cdot \left(\big. i \cdot f_{\Ascr}(G) + f_{\Bscr_0}(G)\right)
\end{equation}
on any instance $(G',\Phi)$, where $G'$ is a subgraph of $G$ and $|I_{\Phi}| \le k_i$.
\end{claim}

We prove the claim by induction on $i$.
By Theorem~\ref{thm:fourApprox} we have a strongly polynomial $\beta_0$-approximation algorithm $\Bscr_0$ for $\Phi$-TSP, 
implying the claim for $i=0$.

Now let $i\in\{1, \dots, i_{\max}\}$.
By our induction hypothesis, there exists a $\beta_{i-1}$-approximation algorithm $\Bscr_{i-1}$ that runs in time $f_{i-1}(G)$
on every weighted graph $G$ and every interface $\Phi$ of $G$  with $|I_{\Phi}| \le k_{i-1}$.
Applying Theorem~\ref{thm:iterImprovement} to the algorithms $\Ascr$ and $\Bscr_{i-1}$ then yields a
$\beta_i$-approximation algorithm $\Bscr_i$ for $\Phi$-TSP that runs on every graph $G$ and every interface $\Phi$ with $|I_{\Phi}|\le k_i$ in time at most 
\begin{align*}
 & |V|^{c\cdot \frac{k_i}{\epsilon'}}
\cdot \left( f_{\mathcal{A}}(G) + f_{\Bscr_{i-1}}\left(G,\, {\textstyle\frac{9\cdot k_{i}}{\epsilon'}} \right)\right) \\
& \le |V|^{c\cdot\frac{k_0}{\epsilon'}} \cdot \left(f_{\mathcal{A}}(G) + f_{\Bscr_{i-1}}\left(G,\, k_{i-1}\right)\right) \\
& \le |V|^{c\cdot\frac{k_0}{\epsilon'}} \cdot \left(\big. f_{\mathcal{A}}(G) + f_{i-1}(G)\right) \\
& = |V|^{c\cdot\frac{k_0}{\epsilon'}} \cdot f_{\mathcal{A}}(G) +  
|V|^{i \cdot c\cdot\frac{k_0}{\epsilon'}} \cdot \left(\big.(i-1) \cdot f_{\Ascr}(G) + f_{\Bscr_0}(G)\right) \\
& \le f_i(G)\enspace.
\end{align*}
\end{proof}

\section{\boldmath Strongly polynomial 4-approximation algorithm for $\Phi$-TSP}
\label{sec:4-Approx}

\fourApprox*

\begin{proof}
Let $\Phi=(I,T,\mathcal{C})$ be an interface of $G=(V,E)$.
By Lemma~\ref{lem:existencePhiTour}, we can assume that $\Phi$ is feasible.
The main component of our algorithm is to obtain a strongly polynomial $2$-approximation algorithm for the problem of finding a set (not a multi-set) 
$F\subseteq E$ of minimum length $\ell(F)$ that satisfies the following three conditions:
\begin{enumerate}[label=(\roman*),topsep=0.3em]
\item\label{item:apxConn} $(V,F)/I$ is connected;
\item\label{item:apxPart} $(V,F)$ connects all vertices within any $C\in \mathcal{C}$;
\item\label{item:apxTEven} each connected component of $(V,F)$ contains an even number of vertices in $T$.
\end{enumerate}
We will achieve this through an application of Jain's iterative rounding method for the Generalized Steiner Network Problem~\cite{jain_2001_factor} together with the elegant framework of Frank and Tardos~\cite{frank_1987_application} to transform certain polynomial-time algorithms into strongly polynomial ones.

Before we discuss the details of Jain's method in our setting together with the framework of Frank and Tardos, we first assume that we can indeed find in strongly polynomial time a set $F\subseteq E$ fulfilling~\ref{item:apxConn},~\ref{item:apxPart}, and~\ref{item:apxTEven} of length no larger than twice the length of a shortest edge set fulfilling these three conditions. 
Because a shortest $\Phi$-tour $\OPT$ must fulfill these conditions, and removing parallel edges does not destroy them, 
there is a subset of $\OPT$ that contains no parallel edges and satisifies~\ref{item:apxConn},~\ref{item:apxPart}, and~\ref{item:apxTEven}.
Therefore, $\ell(F) \leq 2\ell(\OPT)$.

Due to property~\ref{item:apxTEven}, the set $F$ contains a $T$-join $J\subseteq F$, which we can find in linear time through standard techniques. 
We then return $F \cupp (F\setminus J)$, which is indeed a $\Phi$-tour and satisfies
\begin{equation*}
\ell(F\cupp(F\setminus J)) \leq 2\ell(F) \leq 4\ell(\OPT)\enspace,
\end{equation*}
as desired. It remains to show how to obtain a strongly polynomial $2$-approximation algorithm for finding a shortest edge set fulfilling~\ref{item:apxConn},~\ref{item:apxPart}, and~\ref{item:apxTEven}.
We start by showing how an application of Jain's iterative rounding method leads to a polynomial-time, but not necessarily strongly polynomial, $2$-approximation algorithm. 

\medskip

To this end, observe that a set $F\subseteq E$ satisfies~\ref{item:apxConn},~\ref{item:apxPart}, and~\ref{item:apxTEven} if and only if
\begin{equation}\label{eq:genSteinerDescr}
|F\cap \delta(S)| \geq f(S) \qquad \forall S\subseteq V\enspace, 
\end{equation}
where the function $f:2^V\to \{0,1\}$ is defined as follows. For $S\subsetneq V$ with $S\neq \emptyset$, we set $f(S)=1$ if at least one of the following three properties holds:
\begin{enumerate}[label=(\alph*),topsep=0.3em]
\item\label{item:apxFConn} $S\cap I = \emptyset$;
\item\label{item:apxFPart} $\exists\: C\in \mathcal{C}$ s.t.~$S\cap C\neq \emptyset$ and $C\setminus S \neq \emptyset$;
\item\label{item:apxFTEven} $|S\cap T|$ is odd.
\end{enumerate}
Otherwise we set $f(S) = 0$. (In particular, $f(\emptyset)=f(V)=0$.) Indeed, the properties~\ref{item:apxFConn},~\ref{item:apxFPart}, and~\ref{item:apxFTEven} are just reformulations of~\ref{item:apxConn},~\ref{item:apxPart}, and~\ref{item:apxTEven}, respectively.

Jain's technique~\cite{jain_2001_factor} leads to a polynomial-time $2$-approximation algorithm for finding a shortest edge set $F$ satisfying~\eqref{eq:genSteinerDescr} if, first, the function $f$ is \emph{weakly supermodular}, which means
\begin{equation}\label{eq:weaklySupermodular}
f(X) + f(Y) \leq \max\left\{
f(X\cup Y) + f(X\cap Y),\,
f(X\setminus Y) + f(Y\setminus X)
\right\}\qquad \forall X,Y\subseteq V\enspace,
\end{equation}
and, second, one can separate over the polytope
\begin{equation}\label{eq:polSepJain}
P = \left\{x\in [0,1]^E : x(\delta(S)) \geq f(S) \;\,\forall S\subseteq V\right\}\enspace.
\end{equation}
in polynomial time.

We start by showing~\eqref{eq:weaklySupermodular}. Notice that~\eqref{eq:weaklySupermodular} clearly holds if $X\subseteq Y$, because in this case we have $\{X,Y\} = \{X\cup Y, X\cap Y\}$. Hence, in what follows, we always assume that $X\setminus Y \neq\emptyset$ and $Y\setminus X \neq \emptyset$.

Let $f_a$, $f_b$, and $f_c$ be the functions from $2^V$ to $\{0,1\}$ that take a value of $1$ precisely for sets $S\subsetneq V, S\neq \emptyset$ that satisfy~\ref{item:apxFConn},~\ref{item:apxFPart}, or~\ref{item:apxFTEven}, respectively. 
Hence, $f(S) = \max\{f_a(S), f_b(S), f_c(S)\}$. First, one can observe that each of the functions $f_a$, $f_b$, and $f_c$ is weakly supermodular.
Consider first $f_a$ and let $X,Y \subseteq V$ with  $X\setminus Y \neq\emptyset$ and $Y\setminus X \neq \emptyset$.
If $f_a(X) = 1$ then $f_a(X\setminus Y) =1$. Similarly, if $f_a(Y) =1$, then $f_a(Y\setminus X) = 1$. 
Hence, $f_a$ satisfies \eqref{eq:weaklySupermodular}.
The function $f_b$ corresponds to pairwise connectivity requirements and, as shown in~\cite{jain_2001_factor}, is therefore weakly supermodular. 
The function $f_c$ is easily seen to be a so-called \emph{proper} function, which means that $f_c(V)=0$, $f_c$ is symmetric, and $f_c(S_1\cup S_2) \leq \max\{f_c(S_1), f_c(S_2)\}$ for any pair of disjoint sets $S_1, S_2 \subseteq V$. 
Finally, it is well-known that any proper function is weakly supermodular (see~\cite{goemans_1994_improved}).

We say that a set $S\subsetneq V$ with $S\neq\emptyset$ is of type~\ref{item:apxFConn},~\ref{item:apxFPart}, or~\ref{item:apxFTEven}, if it satisfies~\ref{item:apxFConn},~\ref{item:apxFPart}, or~\ref{item:apxFTEven}, respectively.
Because each of the functions $f_a$, $f_b$, and $f_c$ is weakly supermodular, the inequality~\eqref{eq:weaklySupermodular} holds whenever the sets $X$ and $Y$ are of the same type, or if $X$ or $Y$ is none of the three types. Hence, it remains to consider sets $X$ and $Y$ of two different types among the types~\ref{item:apxFConn},~\ref{item:apxFPart}, and~\ref{item:apxFTEven}.
Let $S_a, S_b, S_c\subseteq V$ be sets of type~\ref{item:apxFConn},~\ref{item:apxFPart}, and~\ref{item:apxFTEven}, respectively. Thus, we need to show that~\eqref{eq:weaklySupermodular} holds for the three cases where $(X,Y)$ is either $(S_a,S_b)$, $(S_a,S_c)$, or $(S_b,S_c)$. Moreover, let $C\in \mathcal{C}$ be a set such that $S_b\cap C\neq \emptyset$ and $C\setminus S_b\neq \emptyset$, which exists because $S_b$ is of type~\ref{item:apxFPart}.

We start by considering the case $(X,Y)=(S_a,S_b)$. As discussed, we assume that $S_a\not\subseteq S_b$ and $S_b\not\subseteq S_a$; for otherwise,~\eqref{eq:weaklySupermodular} holds trivially. Notice that in this case we have
\begin{equation*}
2 = f(S_a) + f(S_b) \leq f(S_a\setminus S_b) + f(S_b \setminus S_a) = 2\enspace,
\end{equation*}
because $(S_a\setminus S_b)\cap I \subseteq S_a\cap I = \emptyset$, as well as $(S_b\setminus S_a)\cap I = S_b\cap I$ and $C\subseteq I$. Hence, $S_a\setminus S_b$ is of type~\ref{item:apxFConn} and $S_b\setminus S_a$ is of type~\ref{item:apxFPart}.

Consider now the case $(X,Y)=(S_a,S_c)$. Here, we have
\begin{equation*}
2 = f(S_a) + f(S_c) \leq f(S_a\setminus S_c) + f(S_c\setminus S_a) = 2\enspace,
\end{equation*}
because $(S_a\setminus S_c)\cap I \subseteq S_a\cap I = \emptyset$, implying that $S_a\setminus S_c$ is of type~\ref{item:apxFConn}, and $|(S_c\setminus S_a)\cap T| = |S_c \cap T|$ due to $S_a\cap T \subseteq S_a\cap I = \emptyset$, which implies that $S_c\setminus S_a$ is of type~\ref{item:apxFTEven}.

It remains to consider the case $(X,Y)=(S_b,S_c)$. We first observe that
\begin{align}
\max\left\{f(S_b\setminus S_c),\, f(S_b\cup S_c) 
\right\} &\geq 1\enspace,\text{ and}\label{eq:bcIneq1}\\
\max\left\{
f(S_b\cap S_c),\, f(S_c\setminus S_b)
\right\} &\geq 1\enspace,\label{eq:bcIneq2}
\end{align}
due to the following. Inequality~\eqref{eq:bcIneq1} holds because $S_b\cup S_c$ can be partitioned into $S_c$ and $S_b\setminus S_c$. Because $|S_c\cap T|$ is odd, either $S_b \cup S_c$ or $S_b\setminus S_c$ must also have an odd intersection with $T$ and is thus of type~\ref{item:apxFTEven}. Inequality~\eqref{eq:bcIneq2} follows from an analogous reasoning using the partition of $S_c$ into $S_b\cap S_c$ and $S_c\setminus S_b$.
Moreover, we have
\begin{align}
\max\left\{
f(S_b\setminus S_c),\, f(S_b\cap S_c)
\right\} &\geq 1\enspace, \text{ and}\label{eq:bcIneq3}\\
\max\left\{
f(S_b\cup S_c),\, f(S_c\setminus S_b)
\right\} &\geq 1\enspace, \label{eq:bcIneq4}
\end{align}
because $S_b$ is of type~\ref{item:apxFPart}, i.e., $S_b\cap C\neq\emptyset$ and $C\setminus S_b\neq \emptyset$. Indeed, even without any assumptions on $S_c\subseteq V$, we have that either $S_b\setminus S_c$ or $S_b\cap S_c$ is also of type~\ref{item:apxFPart}. The same holds for either $S_b\cup S_c$ or $S_c\setminus S_b$. Among the four expressions $f(S_b\cup S_c)$, $f(S_b\cap S_c)$, $f(S_b\setminus S_c)$, and $f(S_c\setminus S_b)$, consider any one of minimum value and sum up the two inequalities among~\eqref{eq:bcIneq1},~\eqref{eq:bcIneq2},~\eqref{eq:bcIneq3}, and~\eqref{eq:bcIneq4} containing that expression. This gives the desired result. For example, if $f(S_b\setminus S_c)$ achieves minimum value among the four, then~\eqref{eq:bcIneq1} implies $f(S_b\cup S_c) = 1$ and~\eqref{eq:bcIneq3} implies $f(S_b\cap S_c)=1$. Hence,
\begin{equation*}
2 = f(S_b) + f(S_c) \leq f(S_b\cup S_c) + f(S_b\cap S_c) = 2\enspace,
\end{equation*}
as desired. This completes the proof that $f$ is weakly supermodular.

To apply Jain's method, it remains to show that we can separate over $P$, and we will in fact give a strongly polynomial algorithm.
Given $y\in [0,1]^E$, we will either show that all constraints $y(\delta(S))\geq f(S)$ for $S\subseteq V$ are fulfilled or return one of these constraints that is violated. 
Notice that, because $y\geq 0$, a constraint $y(\delta(S))\geq f(S)$ can only be violated if $f(S)=1$, i.e., $S$ is either of type~\ref{item:apxFConn},~\ref{item:apxFPart},~\ref{item:apxFTEven}.
Hence, we can check these constraints for each type separately.

Whether there is a violated constraints of type~\ref{item:apxFConn} reduces to finding a minimizer of
\begin{equation*}
\min\left\{ y(\delta(S)) : S\subseteq V \text{ with } S\cap I=\emptyset \right\}\enspace.
\end{equation*}
This can be solved through a global minimum cut algorithm applied to the graph $G/I$ with edge weights $y$. Indeed, this either leads to a cut $S$ with $S\cap I=\emptyset$ as desired or one where $I\subseteq S$, in which we can replace $S$ by $V\setminus S$.

To check whether there is a violated constraint of type~\ref{item:apxFPart} reduces to
\begin{equation*}
\min\left\{
y(\delta(S)) : \exists C\in \mathcal{C} \text{ with } S\cap C\neq \emptyset \text{ and } C\setminus S \neq \emptyset
\right\}\enspace.
\end{equation*}
This can be solved by performing the following for all $C\in \mathcal{C}$ with $|C|\geq 2$. Number the vertices in $C$ arbitrarily $C=\{c_1,\ldots, c_k\}$, and solve a minimum $c_i$-$c_{i+1}$ cut problem in $G$ with edge weights $y$ for each $i\in \{1,\ldots, k-1\}$. If any of these $s$-$t$ cut problems leads to a cut of value strictly smaller than $1$, then the minimizing cut corresponds to a violated constraint. Otherwise, there is no violated constraints of type $y(\delta(S))\geq f(S)$ for any set $S$ of type~\ref{item:apxFPart}.

Finally, checking whether there is a violated constraint of type~\ref{item:apxFTEven} reduces to
\begin{equation*}
\min \left\{
y(\delta(S)) : S\subseteq V, |S\cap T| \text{ is odd}
\right\}\enspace.
\end{equation*}
This is a minimum weight $T$-cut problem, for which strongly polynomial algorithms are well known (see, e.g.,~\cite{schrijver_2003_combinatorial}).

In summary, the separation problem over $P$ can be solved in strongly polynomial time, and we can therefore apply Jain's technique as claimed.

It remains to show that the overall algorithm can be transformed into a strongly polynomial one.
This is a consequence of the framework of Frank and Tardos~\cite{frank_1987_application}
(see also \cite{groetschel_1993_geometric,eisenbrand_2010_integer, jain_2001_factor} for similar applications).
More precisely, the only step that is not strongly polynomial in Jain's iterative rounding method is solving linear programs on faces of $P$ with objective function given by $\ell$.
Notice that the coefficients in the constraints describing $P$ are all $0$ or $1$.
Hence, they have small encoding length. For such cases, Frank and Tardos~\cite{frank_1987_application} show how $\ell$ can be replaced (in strongly polynomial time) by another objective $\hat{\ell}$ of encoding length polynomial in the dimension $|E|$ of the problem such that the set of optimal solutions over any polytope in $|E|$ dimensions with constraints of small encoding length is the same for the two objectives $\ell$ and $\hat{\ell}$. 
Hence, one can find an optimal linear programming solution with respect to $\hat{\ell}$ instead of $\ell$, whenever a linear program has to be solved in Jain's procedure.
\end{proof}

\section{Conclusions and open problems}\label{sec:conclusions}

We showed that given a polynomial-time $\alpha$-approximation algorithm for TSP 
we can obtain a polynomial-time $(\alpha + \epsilon)$-approximation algorithm for Path TSP.
Feige and Singh~\cite{feige_2007_improved} proved a similar kind of result for the asymmetric traveling salesman problem (ATSP):
given a polynomial-time $\alpha$-approximation algorithm for ATSP, there is a polynomial-time $(2\alpha + \epsilon)$-approximation algorithm for its path version.
A natural question is whether our techniques can be used to improve on their result and avoid losing a factor of two in the approximation ratio.

For the Asymmetric Path TSP, the relatively simple dynamic program (sketched in Section~\ref{sec:high_level}) still works and 
could be used to reduce to the case where the distance $d(s,t)$ from $s$ to $t$ is not much more than $\sfrac{1}{2}\cdot \ell(\OPT)$.
(We get $\sfrac{1}{2}$ instead of $\sfrac{1}{3}$ because a cut can contain two forward edges and one backward edge, and backward edges can 
 belong to many cuts.)
To make further progress, we might again try to guess edges also in cuts in which $\OPT$ contains a larger, but constant number of edges.
However, even if the distance $d(s,t)$ is very small, the distance $d(t,s)$ from $t$ to $s$ could be large.
In this case we do not know how to reduce to ATSP or guess edges of significant length via dynamic programming.
Another obstacle is the following: Our approach for reducing Path TSP to TSP required a constant-factor approximation algorithm for $\Phi$-TSP.
Thus, for the asymmetric case one would probably need a suitable constant-factor approximation algorithm
for a directed version of $\Phi$-TSP, and we do not know how to obtain this.

A special case of $\Phi$-TSP that is more general than Path TSP is the $T$-tour problem.
Here $I_{\Phi} =T_{\Phi}$ is the given set $T$ and $\Cscr_{\Phi} = \{T\}$.
None of the recent improvements for Path TSP seems to extend to general $T$-tours beyond constant $|T|$,
so Seb\H{o}'s $\sfrac{8}{5}$-approximation~\cite{sebo_2013_eight-fifth} remains the best that we know.
Another question is how well $\Phi$-TSP can be approximated in general.
We showed an approximation ratio of four, but a better ratio might be possible.

\bibliographystyle{plain}

\end{document}